\documentclass{amsart}

\usepackage{mathrsfs}
\usepackage{amsthm}
\usepackage{amsmath}
\usepackage{amssymb}
\usepackage{esint}
\usepackage{amsfonts}
\usepackage{color}

\makeatletter
\theoremstyle{plain}
\newtheorem{thm}{\protect\theoremname}
  \theoremstyle{definition}
  \newtheorem{defn}[thm]{\protect\definitionname}
  \theoremstyle{plain}
  \newtheorem{lem}[thm]{\protect\lemmaname}
   \theoremstyle{plain}
  \newtheorem{prop}[thm]{\protect\propositionname}
  \theoremstyle{remark}
  \newtheorem{rem}[thm]{\protect\remarkname}
  \theoremstyle{remark}
  \newtheorem*{rem*}{\protect\remarkname}
 \theoremstyle{plain}
  \newtheorem{cor}[thm]{\protect\corollaryname}
   
\makeatother

  \providecommand{\definitionname}{Definition}
  \providecommand{\lemmaname}{Lemma}
  \providecommand{\remarkname}{Remark}
\providecommand{\theoremname}{Theorem}
\providecommand{\propositionname}{Proposition}
\providecommand{\corollaryname}{Corollary}

 \renewcommand{\(}{\left(}
\renewcommand{\)}{\right)}

\newcommand{\rr}{\mathop{{\rm I}\mskip -4.0mu{\rm R}}\nolimits}

\newcommand{\eps}{\varepsilon}

\begin{document}
\title[The role of the scalar curvature in    systems on
Riemannian manifolds]{The role of the scalar curvature in some singularly perturbed coupled elliptic systems on
Riemannian manifolds}
\author{Marco Ghimenti}
\address[Marco Ghimenti] {Dipartimento di Matematica,
  Universit\`{a} di Pisa, via F. Buonarroti 1/c, 56127 Pisa, Italy}
\email{marco.ghimenti@dma.unipi.it.}
\author{Anna Maria Micheletti}
\address[Anna Maria Micheletti] {Dipartimento di Matematica,
  Universit\`{a} di Pisa, via F. Buonarroti 1/c, 56127 Pisa, Italy}
\email{a.micheletti@dma.unipi.it.}
\author{Angela Pistoia}
\address[Angela Pistoia] {Dipartimento SBAI, Universit\`{a} di Roma ``La Sapienza", via Antonio Scarpa 16, 00161 Roma, Italy}
\email{pistoia@dmmm.uniroma1.it}

\begin{abstract}
Given a 3-dimensional Riemannian manifold $(M,g)$, we investigate the existence of positive solutions of the  Klein-Gordon-Maxwell system
$$
\left\{ \begin{array}{cc}
-\varepsilon^{2}\Delta_{g}u+au=u^{p-1}+\omega^{2}(qv-1)^{2}u & \text{in }M\\
-\Delta_{g}v+(1+q^{2}u^{2})v=qu^{2} & \text{in }M 
\end{array}\right.
$$
and  Schr\"odinger-Maxwell system 
$$
\left\{ \begin{array}{cc}
-\varepsilon^{2}\Delta_{g}u+u+\omega uv=u^{p-1} & \text{in }M\\
-\Delta_{g}v+v=qu^{2} & \text{in }M
\end{array}\right.
$$
when $p\in(2,6). $  We prove that if $\eps$ is small enough,   any  {\it stable} critical point  $\xi_0$ of the scalar curvature of $g$  generates a positive solution
$(u_\eps,v_\eps)$ 
to both the systems such that $u_\eps$ concentrates  at $\xi_0$  as $\eps$ goes to zero.
\end{abstract}

\date{\today}
 \maketitle

\section{Introduction}
Let $(M,g)$ be a smooth compact, boundaryless $3-$dimensional Riemannian
manifold.

Given real numbers $\varepsilon>0,$ $a>0$, $q>0$, $\omega\in(-\sqrt{a},\sqrt{a})$
and $2<p<6$, we consider the following singularly perturbed electrostatic
Klein-Gordon-Maxwell system
\begin{equation}
\left\{ \begin{array}{cc}
-\varepsilon^{2}\Delta_{g}u+au=u^{p-1}+\omega^{2}(qv-1)^{2}u & \text{in }M\\
-\Delta_{g}v+(1+q^{2}u^{2})v=qu^{2} & \text{in }M\\
u,v>0
\end{array}\right.\label{eq:P}
\end{equation}
and the
Schr\"odinger-Maxwell system

\begin{equation}
\left\{ \begin{array}{cc}
-\varepsilon^{2}\Delta_{g}u+u+\omega uv=u^{p-1} & \text{in }M\\
-\Delta_{g}v+v=qu^{2} & \text{in }M\\
u,v>0
\end{array}\right.\label{eq:P-SM}
\end{equation}

 KGM systems and SM systems provide a model for the
description of the interaction between a charged particle of matter
$u$ constrained to move on $M$ and its own electrostatic field $v$.

The Schr\"odinger-Maxwell and the Klein-Gordon-Maxwell systems have
been object of interest for many authors.

\medskip
In the pioneering paper  \cite{BF4} Benci-Fortunato studied the following
 Schr\"odinger-Maxwell  system
\[
\left\{ \begin{array}{cc}
-\Delta u+u+\omega u v=0& \text{in }\Omega\subset\mathbb{R}^{3}\text{ or in }\mathbb{R}^{3}\\
-\Delta v=\gamma u^{2}& u=v=0\text{ on }\partial\Omega
\end{array}\right.
\]
Regarding the system in a semiclassical regime
\begin{equation} 
\left\{ \begin{array}{cc}
-\varepsilon^{2}\Delta u+u+\omega u v=f(u) & \text{in }\Omega\subset\mathbb{R}^{3}\text{ or in }\mathbb{R}^{3}\\
-\Delta  v=\gamma u^{2} & u=v=0\text{ on }\partial\Omega,
\end{array}\right.\label{eq:DW}
\end{equation}
(here $\eps$ is a positive parameter small enough) Ruiz \cite{R} and D'Aprile-Wei \cite{DW1} showed the existence of a family
of radially symmetric solutions 
respectively for $\Omega=\mathbb{R}^{3}$
  or a ball. D'Aprile-Wei \cite{DW2} also proved the existence
of clustered solutions in the case of a bounded domain $\Omega$ 
in $\mathbb{R}^{3}$. Ghimenti-Micheletti  \cite{GMsm} give an estimate on the number of solutions of \eqref{eq:DW}. 

Moreover, when $\eps=1$ we have results of existence and nonexistence of solutions for 
pure power nonlinearities $f(v)=|v|^{p-2}v$, $2<p<6$ or in presence of a
more general nonlinearity  (see \cite{AR,ADP,AP1,BJL,DM1,IV,Ki,PS,WZ}).
In particular, Siciliano \cite{S} proves an estimate on the number of solution for a pure power nonlinearity when $p$ is close to the 
critical exponent.

\medskip
Klein-Gordon-Maxwell systems are widely studied in physics and in mathematical physics (see
for example  \cite{CB,D,KM,MN,M}). In this setting, there are results of existence
and non existence of solutions for subcritical nonlinear terms in
a bounded domain $\Omega$ (see \cite{AP2,BF1,BF3,C,DM2,DP,DPS1,DPS2,Mu}). 

 As far as we know, the first result concerning the Klein-Gordon systems on manifold is due to Druet-Hebey \cite{DH}. 
They  prove uniform bounds and the existence of a solution for    the system (\ref{eq:P}) when $\eps=1,$  $a$ is positive function 
and the exponent $p$ is either  subcritical or critical, i.e. $p\in(2,6].$ In particular, the existence of a solution  in the critical case, i.e. $p=6$,
is obtained provided the function $a$ is suitable   small with respect to the scalar curvature of the metric $g.$
 Recently, Ghimenti-Micheletti \cite{GMkg} give an estimate on the number of low energy solution for the system (\ref{eq:P}) 
in terms of the topology of the manifold. 

\medskip
In this paper, we show that the existence and the multiplicity of solutions of both     systems (\ref{eq:P}) and   (\ref{eq:P-SM})
 in the subcritical case when $\eps$ is small enough
is strictly related to   the geometry of   the manifold $(M,g).$ More precisely, we prove that
 the number  of solutions  to (\ref{eq:P}) or   (\ref{eq:P-SM})
is affected by the number of {\it stable} critical points of the scalar curvature  $S_{g}$  of the metric $g$. 
Indeed, our result reads as follows.
\begin{thm}
\label{main} Assume $K$ is a $C^{1}$-stable critical set of
  $S_{g}.$  
Then there exists $\bar{\varepsilon}>0$ such that for any $\varepsilon\in(0,\bar{\varepsilon})$
the KGM system (\ref{eq:P}) and the SM system (\ref{eq:P-SM}) have a solution $(u_{\varepsilon},v_{\varepsilon})$
such that $u_{\varepsilon}$   concentrates at a point $\xi_0\in K$
as $\varepsilon$ goes to $0$.  More precisely, there exists a point $\xi_\eps\in M$ such that if $\eps$ goes to zero $\xi_\eps\to\xi_0\in K$ 
$$u_\eps- W_{\eps,\xi_\eps}\to 0 \ \text{in}\ H^1_g(M),\ v_\eps \to 0 \ \text{in}\ H^1_g(M)$$
where the function $W_{\eps,\xi_\eps}$ is defined in \eqref{w1}.
  
\end{thm}

We recall the  the definition
of $C^{1}$-stable critical set.
\begin{defn}\label{c1stable}
Let $f\in C^{1}(M,\mathbb{R})$. We say that $K\subset M$ is a $C^{1}$-stable
critical set of $f$ if $K\subset\left\{ x\in M\ :\ \nabla_{g}f(x)=0\right\} $
and for any $\mu>0$ there exists $\delta>0$ such that, if $h\in C^{1}(M,\mathbb{R})$
with
\[
\max_{d_{g}(x,K)\le\mu}|f(x)-h(x)|+|\nabla_{g}f(x)-\nabla_{g}h(x)|\le\delta,
\]
then $h$ has a critical point $\xi$ with $d_{g}(\xi,K)\le\mu$.
Here $d_{g}$ denotes the geodesic distance associated to the Riemannian
metric $g$.\end{defn}

It is easy to see that if $K$ is the set of the strict local minimum
(or maximum) points of $f$, then $K$ is a $C^{1}$-stable critical
set of $f$. Moreover, if $K$ consists of nondegenerate critical
points, then $K$ is a $C^{1}$-stable critical set of $f$.

By Theorem \ref{main} we deduce that multiplicity of solutions of (\ref{eq:P})
and (\ref{eq:P-SM}) is strictly related to stable critical points
of the scalar curvature. At this aim, it is useful to point out that Micheletti-Pistoia \cite{MP2} proved that, 
generically with respect to the metric $g$, the scalar curvature $S_g$ is a Morse function on the manifold $M.$ More precisely, they proved  
\begin{thm} Let $\mathscr{M}^{k}$ be the set
of all $C^{k}$ Riemannian metrics on $M$ with $k\ge3$. 
The set
\[
\mathscr{A}=\left\{ g\in\mathscr{M}^{k}\ :\ \text{all the critical points of }S_{g}\text{ are non degenerate}\right\}
\]
is a open dense subset of $\mathscr{M}^{k}$.  
\end{thm}
Then generically with respect to the metric $g$, the  critical points of the scalar curvature $S_g$ 
 are nondegenerate, in a finite number and at least $  P_{1}(M)$
where $P_{t}(M)$ is the Poincar\'e polynomial of $M$ in the $t$ variable. Therefore, we can conclude as follows.

\begin{cor} Generically with respect to the metric $g,$ if $\eps$ is small enough the KGM system (\ref{eq:P}) and the SM system (\ref{eq:P-SM})  
 have at
least $P_{1}(M)$  positive solutions $(u_{\varepsilon},v_{\varepsilon})$
such that $u_{\varepsilon}$
concentrates at one nondegenerate critical point of the scalar curvature $S_{g}$ as $\varepsilon$
goes to $0$.\end{cor}

\medskip
The proof of our results relies on a very well known Ljapunov-Schmidt reduction.
 In Section \ref{var-set} we recall some known results, we write the approximate solution, we sketch  the proof of the Ljapunov Schmidt procedure and we prove Theorem \ref{main}. In Section \ref{redux} we reduce the problem to a finite dimensional one, while in  Section \ref{red-exp} we study the reduced problem.
In  Appendix A we give some important estimates.
All the proofs are given for the system (\ref{eq:P}), but it is clear that up some minor modifications they also hold true  for the system
  (\ref{eq:P-SM}).

\section{Preliminaries and scheme of the proof of Theorem \ref{main}} \label{var-set}
\subsection{The function $\Psi$}
First of all, we reduce the system to a single equation.
In order to overcome the problems given by the competition between
$u$ and $v$, using an idea of Benci and Fortunato \cite{BF1}, we introduce
the map $\Psi:H_{g}^{1}(M)\rightarrow H_{g}^{1}(M)$ defined by the
equation
\begin{equation}
-\Delta_{g}\Psi(u)+\Psi(u)+q^{2}u^{2}\Psi(u)=qu^{2}.\label{eq:e1}
\end{equation}
It follows from standard variational arguments that $\Psi$ is well-defined in $H^1_g(M)$ as soon
as $\lambda:=  a-\omega^2>0,$ i.e.  $\omega\in]-\sqrt a,+\sqrt a[.$

By the maximum principle and by regularity theory is not difficult
to prove that
\begin{equation}\label{psipos}
0<\Psi(u)<1/q\ \hbox{for all $u$ in $H_{g}^{1}(M)$}.
\end{equation}
Moreover, it holds true that
\begin{lem}
\label{lem:e1}The map $\Psi:H_{g}^{1}(M)\rightarrow H_{g}^{1}(M)$
is $C^{1}$ and its differential $\Psi'(u)[h]=V_{u}[h]$ at $u$ is
the map defined by
\begin{equation}
-\Delta_{g}V_{u}[h]+V_{h}[h]+q^{2}u^{2}V_{u}[h]=2qu(1-q\Psi(u))h\text{ for all }h\in H_{g}^{1}(M).\label{eq:e2}
\end{equation}
Also, we have
\[
0\le\Psi'(u)[u]\le\frac{2}{q}
\]

\end{lem}

\begin{lem}
\label{lem:e2}The map $\Theta:H_{g}^{1}(M)\rightarrow\mathbb{R}$
given by
\[
\Theta(u)=\frac{1}{2}\int_{M}(1-q\Psi(u))u^{2}d\mu_{g}
\]
is $C^{1}$ and
\[
\Theta'(u)[h]=\int_{M}(1-q\Psi(u))^{2}uhd\mu_{g}\ \hbox{for any $u,h\in H_{g}^{1}(M)$}
\]

\end{lem}
For the proofs of these results we refer to \cite{DH}.

Now, we introduce the functionals   $I_\eps,J_\eps,G_\eps:{\rm H}^1_g(M)\rightarrow\rr$  
\begin{equation}\label{ieps}
I_\eps(u)=J_\eps(u)+{\omega^2\over2}G_\eps(u),\end{equation}
where
\begin{equation}\label{jieps}
J_\eps(u):={1\over\eps^3}\int\limits_M\left[{1\over2}\eps^2|\nabla_gu|^2+
{\lambda\over2}  u ^2-F\left(u \right) \right]d\mu_g\end{equation}
and
\begin{equation}\label{geps}
G_\eps(u):={1\over\eps^3}q\int\limits_M\Psi(u) u^2 d\mu_g.\end{equation}
Here $F(u):={1\over p} \left(u^+\right)^p,$ so that $F'(u)=f(u):=(u^+)^{p-1}.$
By Lemma \ref{lem:e2} we deduce that
\[
\frac{1}{2}G_{\varepsilon}'(u)[\varphi]=\frac{1}{\varepsilon^{3}}\int_{M}[q^{2}\Psi^{2}(u)-2q\Psi(u)]u\varphi d\mu_{g},
\]
 so
\[
I'_{\varepsilon}(u)\varphi=\frac{1}{\varepsilon^{3}}\int_{M}\varepsilon^{2}\nabla_{g}u\nabla_{g}\varphi+au\varphi-(u^{+})^{p-1}\varphi-\omega^{2}(1-q\Psi(u))^{2}u\varphi d\mu_{g}.
\]

Therefore, if $u$ is a critical point of the functional $I_\varepsilon$ we have

\begin{equation}\label{p}
-\varepsilon^{2}\Delta_{g}u+(a-\omega^2) u+\omega^2q\Psi(u)(2-q\Psi(u))u=(u^+)^{p-1}\ \   \text{ in }M.
\end{equation}
In particular, if $u\ne0$ by the maximum principle and by the regularity theory we have that $u>0$.
Thus the pair $(u,\Psi(u))$ is a solution of Problem (\ref{eq:P}).
Finally, the problem is reduced to find a solution to the single equation \eqref{p}.

\subsection{Setting of the problem}

In the following we denote by  $B_g(\xi,r)$   the geodesic ball in $M$ centered in $\xi$ with radius $r$
and by  $B (x,r)$ the ball in
$\mathbb{R}^{3}$ centered in $x$ with radius
$r$.    

 It is possible to define a system of coordinates on $M$ called normal coordinates. We denote by $g_{\xi}$   the Riemannian metric read in $B(0,r)\subset \rr^3$ through the normal coordinates defined by the exponential map
${\rm exp}_{\xi}$ at $\xi.$ We denote $|g_{\xi}(z)|:={\rm det}\left(g_{ij}(z)\right)$ and $\left(g^{ij}_{\xi}(z)\right)$
is the inverse matrix of $g_{\xi}(z).$
In particular, it holds
\begin{equation}\label{f2}
  g^{ij}_{\xi}(0)=\delta_{ij}\ \text{and}\ {\partial g^{ij}_{\xi}\over\partial z_k}(0)=0\ \hbox{for any $i,j,k.$}\end{equation}
Here $\delta_{ij}$ denotes the Kronecker symbol.

We denote by 
$$\|u\|^2_g:=\int\limits_M\( |\nabla_gu|^2+u^2\)d\mu_g\ \text{and}\ |u|^q_g:=\int\limits_M  |u|^q d\mu_g$$
the standard norms in the spaces
$H^1_g(M)$ and $L^q(M).$

Let $H_\eps$ be the Hilbert space
${\rm H}^1_g(M)$ equipped with the inner product
$$\left\langle u,v\right\rangle _\eps:={1\over\eps^3}\left(\eps^2\int\limits_M \nabla_gu\nabla_gvd\mu_g+\lambda\int\limits_M uvd\mu_g\right),$$
which induces the norm
$$\|u\|^2_\eps:={1\over\eps^3}\left(\eps^2\int\limits_M |\nabla_g|^2d\mu_g+\lambda\int\limits_M u^2d\mu_g\right).$$
Let $L^q_\eps$ be the Banach space
${\rm L}^q_g(M)$ equipped   the norm
$$ |u | _{q,\eps}:=\left({1\over\eps^3} \int\limits_M |u|^qd\mu_g\right)^{1/q} .$$

It is clear that for any $q\in[2,6)$  the embedding $H_\eps\hookrightarrow L^q_\eps$
is a continuous map.
It is not difficult to check that
\begin{equation}\label{r2}
| u   |_{q,\eps}\le c \left\|u\right\|_\eps ,\ \hbox{for any}\ u\in H_\eps,
\end{equation}
where the constant $c$ does not depend on $\eps.$

 In particular, the embedding $i_\eps:H_\eps\hookrightarrow L^p_\eps$ is a compact continuous map. The adjoint
 operator $i^*_\eps: L^{p'}_\eps\rightarrow H_\eps,$ $p':={p\over p-1},$ is a   continuous map such that
 $$u=i^*_\eps(v)\ \Leftrightarrow\ \left\langle i^*_\eps(v),\varphi\right\rangle_\eps={1\over\eps^3}\int\limits_M v\varphi,
 \ \varphi\in H_\eps\ \ \Leftrightarrow\
 -\eps^2\Delta_g u+u=v\ \text{on}\  M, \  u\in{\rm H}^1_g(M).  $$
  Moreover
  $$\left\|i^*_\eps(v)\right\|_\eps\le c  | v   |_{p',\eps} ,\ \hbox{for any}\ v\in L^{p'}_\eps,$$
where the constant $c$ does not depend on $\eps.$

We can rewrite problem (\ref{p}) in the equivalent way
\begin{equation}\label{p1}
u=i^*_\eps\left [f(u) +\omega^2 g(u)\right],\ u\in H_\eps,\end{equation}
where we set
\begin{equation}\label{go}
g(u):= \left(q^2\Psi^2(u) -2q\Psi(u)\right)u .\end{equation}

\subsection{An approximation for the solution}

 It is well known (see  \cite{GNN,K})  that there exists a unique positive spherically symmetric function $U\in{\rm H}^1 (\rr^N)$ such that
\begin{equation}\label{pl}
-\Delta U+\lambda U=U^{p-1}\ \hbox{in}\ \rr^N.\end{equation}
Moreover, the function $U$ and its derivatives are   exponentially decaying at infinity, namely
\begin{equation}\label{pll}
\lim\limits_{|x|\rightarrow\infty}U(|x|)|x|^{N-1\over2}e^{|x|}=c >0,\
 \lim\limits_{|x|\rightarrow\infty}U'(|x|)|x|^{N-1\over2}e^{|x|}=-c .
\end{equation}

Let $\chi_r$ be  a smooth cut-off function  such that $\chi_r(z)=1$ if $z\in B(0,r/2) ,$
$\chi_r(z)=0$ if $z\in \rr^3\setminus B(0,r ) ,$   $|\nabla \chi _r(z)|\le 2/r $ and
$|\nabla^2 \chi _r(z)|\le 2/r^2 $ where $r$ is the injectivity radius of $M.$
Fixed a point $\xi\in
M$ and $\eps>0$ let us define on $M$ the function
\begin{equation}\label{w1}
 W_{\eps,\xi}(x):= U_\eps\left({\rm exp}^{-1}_\xi(x)\right)\chi_r\left( {\rm exp}^{-1}_\xi(x) \right)
  \ \text{if}\ x\in B_g(\xi,r) ,\  W_{\eps,\xi}(x):=
0 \ \text{if}\ x\in M\setminus B_g(\xi,r) ,
\end{equation}
where we set $U_\eps(z):=U\left({z\over\eps}\right).$

We will look for a solution to (\ref{p1}) or equivalently to (\ref{p}) as $u_\eps:= W_{\eps,\xi}+\phi,$
where the rest term $\phi$ belongs to a suitable space which will be introduced in the following.

It is well known that every solution to the linear equation
$$-\Delta \psi+\lambda\psi=(p-1)U^{p-2}\psi\ \hbox{in}\ \rr^3$$
is a linear combination of the functions
$$\psi^i(z):={\partial U\over\partial z_i}(z),\ i=1,2,3.$$

Let us define on $M$ the functions
\begin{equation}\label{zi}
 Z^i_{\eps,\xi}(x):= \psi^i_\eps\left({\rm exp}^{-1}_\xi(x)\right)\chi_r\left( {\rm exp}^{-1}_\xi(x)\right)
 \ \text{if}\ x\in B_g(\xi,r) ,\  Z^i_{\eps,\xi}(x):=
0 \ \text{if}\ x\in M\setminus B_g(\xi,r) ,
\end{equation}
where we set $\psi^i_\eps(z):=\psi^i\left({z\over\eps}\right).$
Let us introduce the spaces
 $K_{\eps,\xi}:={\rm span}\left\{ Z^1_{\eps,\xi},Z^2_{\eps,\xi}, Z^3_{\eps,\xi}\right\}$
and\break
 $K^\perp_{\eps,\xi}:= \left\{\phi\in H_\eps\ :\ \left\langle \phi, Z^i_{\eps,\xi}\right\rangle_\eps=0,\ i=1,2,3\right\}.$
Finally, let
 $\Pi_{\eps,\xi}:H_\eps\rightarrow K_{\eps,\xi}$ and $ \Pi^\perp_{\eps,\xi}:H_\eps\rightarrow K^\perp_{\eps,\xi}$
be the orthogonal projections.

In order to solve problem (\ref{p1}) we will solve the couple of equations
\begin{eqnarray}\label{p2}
& &\Pi^\perp_{\eps,\xi}\left\{W_{\eps,\xi}+\phi-
i^*_\eps\left[f\left(W_{\eps,\xi}+\phi\right)+\omega^2g\left(W_{\eps,\xi}+\phi\right)\right]\right\}=0
\\ & &\nonumber\\ 
& &\Pi _{\eps,\xi}\left\{W_{\eps,\xi}+\phi-
i^*_\eps\left[f\left(W_{\eps,\xi}+\phi\right)+\omega^2g\left(W_{\eps,\xi}+\phi\right)\right]\right\}=0.\label{p3}\end{eqnarray}

\subsection{Scheme of the proof of Theorem \ref{main}}

The first step is to solve equation \eqref{p2}. More precisely,  if  $\eps  $ is small enough for any fixed $\xi \in M$, we will find a function
$\phi \in \Pi^\perp_{\eps,\xi}$ such that \eqref{p2} holds.

First of all, we define the linear operator $L_{\eps,\xi}:K^\perp_{\eps,\xi}\rightarrow K^\perp_{\eps,\xi}$
by
$$L_{\eps,\xi}(\phi):=\Pi _{\eps,\xi}^\perp\left\{\phi-i^*_\eps\left[f'\left(W_{\eps,\xi}\right) \phi\right]\right\}.$$

In Proposition 3.1 of \cite{MP} we proved the invertibility of $L_{\eps,\xi}.$

\begin{prop}\label{inv}
There exists $\eps_0>0$ and $c>0$ such that for any $\xi\in M$ and for any $\eps\in(0,\eps_0)$
$$\left\|L_{\eps,\xi} (\phi)\right\|_\eps \ge c\|\phi\|_\eps\ \hbox{ for any $\phi\in K^\perp_{\eps,\xi}$}.$$
\end{prop}

Secondly, in Lemma 3.3 of \cite{MP}  we    estimated the error term   $R_{\eps,\xi} $ defined by
 \begin{equation}\label{re}
 R_{\eps,\xi}:=\Pi^\perp_{\eps,\xi}\left\{i^*_\eps\left[f \left(W_{\eps,\xi}\right)\right]-W_{\eps,\xi}\right\}.
 \end{equation}

\begin{prop}\label{re1}
There exists $\eps_0>0$ and $c>0$  such that for any $\xi\in M$ and for any $\eps\in(0,\eps_0)$ it holds
 $$  \left\|  R_{\eps,\xi}\right\|_\eps\le c\eps^2.$$
\end{prop}

 Finally, we use a   contraction mapping argument   to solve equation \eqref{p2}. This is done in Section \ref{redux}

\begin{prop}\label{resto}
There exists $\eps_0>0$ and $c>0$ such that for any $\xi\in M$ and for any $\eps\in(0,\eps_0)$
there exists a unique $\phi_{ \eps,\xi}=\phi(\eps,\xi)$ which solves equation (\ref{p2}). Moreover
\begin{equation}\label{resto1}
  \left\|\phi_{ \eps,\xi}\right\|_\eps\le c\eps^2.\end{equation} Finally,
   $\xi\rightarrow \phi_{ \eps,\xi}$ is a $C^1-$map.
\end{prop}

\medskip
The second step is to solve equation \eqref{p3}. More precisely, for $\eps$ small enough we will find the point $\xi$ in $M$ such that equation \eqref{p3} is satisfied.

Let us introduce the reduced energy  $\widetilde I_\eps:M\rightarrow\rr$    defined by
$$\widetilde I_\eps(\xi):=I_\eps\left(W_{\eps,\xi}  +\phi_{\eps,\xi}\right),$$
where the energy $I_\eps$ whose critical points are solution to problem  (\ref{p}) is defined in \eqref{ieps}.

First of all,  arguing   exactly as in Lemma 4.1 of \cite{MP} we get
\begin{prop}\label{ridotto}
  $\xi_\eps$ is a critical point of $\widetilde I_\eps  $ if and only if
the function $u_\eps=W_{\eps,\xi_\eps}  +\phi_{\eps,\xi_\eps}$ is a solution to problem (\ref{p}).
\end{prop}

Thus, the problem is reduced to search for critical points of $\widetilde I_\eps  $ whose asymptotic expansion is given in Section \ref{red-exp}
and reads as follows.

\begin{prop}\label{espa}
It holds true that
\begin{equation}\label{e1}
  \widetilde I_\eps\left( \xi\right)
=c_1+c_2\eps^2 -c_3 S_g(\xi) \eps^2 +o\left(\eps^2\right)
,\end{equation}
 $C^1-$uniformly with respect to
$\xi $  as $\eps$ goes to zero.
 Here $S_g(\xi)$ is the scalar curvature of $M$ at $\xi $ and $c_i$'s are constants.
 \end{prop}

Finally, we can prove  Theorem  \ref{main} by
showing that $ \widetilde I_\eps$ has a critical point in $M$.

\begin{proof}[Proof of Theorem \ref{main}]
If  $K$ is a $C^1$-stable critical set of     the scalar curvature of $M$ (see Definition \ref{c1stable}),
 by Proposition \ref{espa}, we deduce that if $\eps  $
is small enough the function $\widetilde I_\eps$ has a critical point $\xi_\eps$ such that $\xi_\eps\to\xi_0$ as $\eps$ goes to zero. The claim follows by
Proposition \ref{ridotto}.
\end{proof}

\section{The finite dimensional reduction}\label{redux}

This section is devoted to the proof of Proposition  \ref{resto}.

First, we remark that equation (\ref{p2}) is equivalent to
\begin{equation}\label{i1}
L_{\eps,\xi}(\phi)=N_{\eps,\xi}(\phi)+S_{\eps,\xi}(\phi)+R_{\eps,\xi},\end{equation}
where
\begin{equation}\label{ne}
 N_{\eps,\xi}(\phi):=\Pi^\perp_{\eps,\xi}\left\{
 i^*_\eps\left[f \left(W_{\eps,\xi}+\phi\right)-f \left(W_{\eps,\xi}\right)-f' \left(W_{\eps,\xi}\right)
 \phi\right] \right\},
 \end{equation}
\begin{equation}\label{se}
 S_{\eps,\xi}(\phi):=\omega^2\Pi^\perp_{\eps,\xi}\left\{i^*_\eps\left[q^2\Psi^2 \left(W_{\eps,\xi}+\phi\right)-
 2q \Psi  \left(W_{\eps,\xi}+\phi\right)\left(W_{\eps,\xi}+\phi\right)\right] \right\}
 \end{equation}
 and $R_{\eps,\xi}$ is defined in \eqref{re}.
In order to solve equation  (\ref{i1}), we need to find
a fixed point for the operator $T_{\eps,\xi}:K^\perp_{\eps,\xi}\rightarrow K^\perp_{\eps,\xi}$
defined by
$$T_{\eps,\xi}(\phi):=L_{\eps,\xi}^{-1}\left(N_{\eps,\xi}(\phi)+S_{\eps,\xi}(\phi)+R_{\eps,\xi}\right).$$
We are going to prove that $T_{\eps,\xi}$ is a contraction map on suitable ball of $H_\eps.$

In Proposition \ref{re1} we estimate the error term  $R_{\eps,\xi}$, while
in Proposition 3.5 of \cite{MP}, we   estimated the higher order term  $N_{\eps,\xi}(\phi)$.
\begin{lem}\label{ne1}
There exists $\eps_0>0,$  $c>0$  such that for any $\xi\in M,$   $\eps\in(0,\eps_0)$ and $r>0$ it holds true that
$$\|N_{\eps,\xi}(\phi)\|_\eps\le c \|\phi\|^2_\eps\ \text{if}\ p\ge3,\quad \|N_{\eps,\xi}(\phi)\|_\eps\le c \|\phi\|^{p-1}_\eps\ \text{if}\ 2<p<3$$
and
$$\|N_{\eps,\xi}(\phi_1)-N_{\eps,\xi}(\phi_2)\|_\eps\le \ell_\eps \|\phi_1-\phi_2\| _\eps $$
provided $\phi,\phi_1,\phi_2\in\left\{\phi\in H_\eps\ :\ \|\phi\|_\eps\le r\eps^2\right\}.$ Here $\ell_\eps\rightarrow0$ while $\eps\rightarrow0$
\end{lem}

It only remains to estimate the term  $S_{\eps,\xi}(\phi)$.

\begin{lem}\label{se1}
There exists $\eps_0>0,$  $c>0$  such that for any $\xi\in M,$   $\eps\in(0,\eps_0)$ and $r>0$ it holds true that
\begin{equation}\label{se11}\|S_{\eps,\xi}(\phi)\|_\eps\le c  \eps^2\end{equation}
and
\begin{equation}\label{se12}\|S_{\eps,\xi}(\phi_1)-S_{\eps,\xi}(\phi_2)\|_\eps\le \ell_\eps \|\phi_1-\phi_2\| _\eps \end{equation}
provided $\phi,\phi_1,\phi_2\in\left\{\phi\in H_\eps\ :\ \|\phi\|_\eps\le r\eps^2\right\}.$ 
Here $\ell_\eps\rightarrow0$ while $\eps\rightarrow0$

\end{lem}
\begin{proof}
Let us prove \eqref{se11}.
By Remark 2.2 in \cite{MP} it follows that
$$\|S_{\eps,\xi}(\phi)\|_\eps\le c\left|\Psi^2\(W_{\eps,\xi}+\phi\)\(W_{\eps,\xi}+\phi\)\right|_{\eps,p'}+\left|\Psi\(W_{\eps,\xi}+\phi\)\(W_{\eps,\xi}+\phi\)\right|_{\eps,p'}.$$
By Lemma \ref{lem:e3} we have
\begin{align*}
&\left|\Psi\(W_{\eps,\xi}+\phi\)\(W_{\eps,\xi}+\phi\)\right|_{\eps,p'}\\&=\({1\over\eps^3}\int\limits_M \left|\Psi\(W_{\eps,\xi}+\phi\)\right|^{p'}\left|W_{\eps,\xi}+\phi \right|^{p'}\)^{1/p'}\\
&\le c{1\over\eps^{3/p'}}\(\int\limits_M \left|\Psi\(W_{\eps,\xi}+\phi\)\right|^{p}\)^{1/p}\(\left|W_{\eps,\xi}+\phi \right|^{p\over p-2}\)^{p-2\over p}\\
\\
&\le c{1\over\eps^{3/p'}}  \left\|\Psi\(W_{\eps,\xi}+\phi\)\right\|_{H^1_g} \left |W_{\eps,\xi}+\phi \right |_{{p\over p-2},g}\\
&\le c{1\over\eps^{3/p'}}  \(\eps^{5\over2}+\|\phi\| ^2_{H^1_g}\)  \(\eps^{3{p\over p-2}}+\|\phi\| _{H^1_g}\)\ \(\hbox{because $\|u\|_{H^1_g}\le \sqrt \eps\|u\| _\eps$}\)
\\ &\le c{1\over\eps^{3/p'}}  \(\eps^{5\over2}+\eps\|\phi\| ^2_\eps\)  \(\eps^{3{p\over p-2}}+\sqrt\eps\|\phi\| _\eps\)\ \(\hbox{because $\|\phi\|_\eps\le r\eps^2$}\)\\
&\le  c{\eps^5\over\eps^{3/p'}}.
\end{align*}
 By Lemma \ref{lem:e3} and the previous estimate we deduce the following
\begin{align*}
&\left|\Psi^2\(W_{\eps,\xi}+\phi\)\(W_{\eps,\xi}+\phi\)\right|_{\eps,p'}\le  c{\eps^2},
\end{align*}
and then \eqref{se11} follows.

Let us prove \eqref{se12}.
By Remark 2.2 in \cite{MP} it follows that
\begin{align*}
&\|S_{\eps,\xi}(\phi_1)-S_{\eps,\xi}(\phi_2)\|_\eps\\ &\le
c\left|\left[\Psi \(W_{\eps,\xi}+\phi_1\)-\Psi \(W_{\eps,\xi}+\phi_2\)\right] W_{\eps,\xi} \right|_{\eps,p'}
+c\left|\left[\Psi ^2\(W_{\eps,\xi}+\phi_1\)-\Psi ^2\(W_{\eps,\xi}+\phi_2\)\right] W_{\eps,\xi} \right|_{\eps,p'}\\
&+c\left|\Psi\(W_{\eps,\xi}+\phi_1\) \phi_1-\Psi\(W_{\eps,\xi}+\phi_2\) \phi_2\right|_{\eps,p'}
+c\left|\Psi\(W_{\eps,\xi}+\phi_1\) \phi_1-\Psi\(W_{\eps,\xi}+\phi_2\) \phi_2\right|_{\eps,p'}\\
&=:I_1+I_2+I_3+I_4. \end{align*}

By Remark \ref{remark:Weps} and Lemma \ref{lem:e7} we have for some $\theta\in(0,1)$:
\begin{align*}
&I_1^{p'}= {1\over\eps^3}\int\limits_M \left|\Psi' \(W_{\eps,\xi}+\theta\phi_1 +(1-\theta)\phi_2\)(\phi_1-\phi_2)\right|^{p'}\left|W_{\eps,\xi}  \right|^{p'} \\
&\le {1\over\eps^3}\(\int\limits_M \left|\Psi' \(W_{\eps,\xi}+\theta\phi_1 +(1-\theta)\phi_2\)(\phi_1-\phi_2)\right|^{p}\)^{p'/p}
\(\int\limits_M \left|W_{\eps,\xi}  \right|^{p'p\over p-p'}\)^{p-p'\over p} \\
&\le c { \eps ^{3{p-p'\over p} }\over\eps^3}\( \eps^2+\|\phi_1\|_{H^1_g}+\|\phi_2\|_{H_g^1}\)^{p'}\|\phi_1-\phi_2\|^{p'}_{H^1_g} \\
  & \(\hbox{because  $\|\phi_i\| _{H^1_g}\le \sqrt\eps\|\phi_i\| _{\eps}\le c\eps^{5/2} $}\)\\
&\le   c   \eps ^{\(2-{3\over p}\)p'} \(\sqrt \eps\)^{p'} \|\phi_1-\phi_2\|^{p'}_{\eps}.
\end{align*}
 By Lemma \ref{lem:e3} and the estimate of $I_1$ we have
\begin{align*}
&I_2^{p'}= {1\over\eps^3}\int\limits_M \left|\Psi  \(W_{\eps,\xi}+\theta\phi_1  \)+\Psi  \(W_{\eps,\xi}+\theta\phi_2  \) \right|^{p'}\left|\Psi  \(W_{\eps,\xi}+\theta\phi_1  \)-\Psi  \(W_{\eps,\xi}+\theta\phi_2  \) \right|^{p'} \left|W_{\eps,\xi}  \right|^{p'}\\
&\le c  \(  \eps^{3/2}+\|\phi_1\|^2_{H^1_g}+\|\phi_2\|^2_{H_g^1}\)^{p'}\eps ^{\(2-{3\over p}\)p'} \(\sqrt \eps\)^{p'} \|\phi_1-\phi_2\|^{p'}_{\eps}\\
  & \(\hbox{because  $\|\phi_i\| _{H^1_g}\le \sqrt\eps\|\phi_i\| _{\eps}\le c\eps^{5/2} $}\)\\
&\le c     \eps^{{3\over 2}p'}\eps ^{\(2-{3\over p}\)p'} \(\sqrt \eps\)^{p'} \|\phi_1-\phi_2\|^{p'}_{\eps},
\end{align*}

 By Lemma \ref{lem:e3} and Lemma \ref{lem:e7}   we have
\begin{align*}
&I_3^{p'}\le {1\over\eps^3}\int\limits_M \left|\Psi' \(W_{\eps,\xi}+\theta\phi_1 +(1-\theta)\phi_2\)(\phi_1-\phi_2)\right|^{p'}\left|\phi_1  \right|^{p'}
+{1\over\eps^3}\int\limits_M \left|\Psi  \(W_{\eps,\xi}+\theta\phi_ 2\) \right|^{p'}\left|\phi_1-\phi_2  \right|^{p'} \\
&\le c {1\over\eps^3}\(\int\limits_M \left|\Psi' \(W_{\eps,\xi}+\theta\phi_1 +(1-\theta)\phi_2\)(\phi_1-\phi_2)\right|^{p}\)^{p'/p}
\(\int\limits_M \left|\phi_1 \right|^{p'p\over p-p'}\)^{p-p'\over p} \\ &
+ c{1\over\eps^3}\(\int\limits_M \left|\phi_1-\phi_2  \right|^{p}\)^{p' /p}
\(\int\limits_M \left|\Psi  \(W_{\eps,\xi}+ \phi_2\) \right|^{p'p\over p-p'}\)^{p-p'\over p}
\\
&\le c {1\over\eps^3}\( \eps^2+\|\phi_1\|_{H^1_g}+\|\phi_2\|_{H_g^1}\)^{p'}\|\phi_1-\phi_2\|^{p'}_{H^1_g}\|\phi_1 \|^{p'}_{H^1_g}
+c {1\over\eps^3}\( \eps^{5/2}+\|\phi_2\|^2_{H^1_g} \)^{p'}\|\phi_1-\phi_2\|^{p'}_{H^1_g}
\\ & \(\hbox{because  $\|\phi_i\| _{H^1_g}\le \sqrt\eps\|\phi_i\| _{\eps}\le c\eps^{5/2} $}\)\\
&\le   c\eps^{{5\over 2}p'-3}\|\phi_1-\phi_2\|_\eps
\end{align*}
 and
 \begin{align*}
&I_4^{p'}\le {1\over\eps^3}\int\limits_M \left|\Psi^2 \(W_{\eps,\xi}+ \phi_1 \)\right|^{p'} \left|\phi_1-\phi_2 \right|^{p'}
\\ &+ {1\over\eps^3}\int\limits_M\left|\Psi  \(W_{\eps,\xi}+ \phi_1 \)+\Psi  \(W_{\eps,\xi}+ \phi_2 \) \right|^{p'} \left|\Psi ' \(W_{\eps,\xi}+\theta\phi_1 +(1-\theta)\phi_2\)(\phi_1-\phi_2)\right|^{p'}\left|\phi_2  \right|^{p'}
 \\ &\le c {1\over\eps^3} \( \eps^{3/2}+\|\phi_1\|^2_{H^1_g}\)^{p'} \( \eps^{5/2}+\|\phi_1\|^2_{H^1_g}\)^{p'} \|\phi_1-\phi_2 \|^{p'}_{H^1_g}
\\ &+ c {1\over\eps^3} \( \eps^{3/2}+\|\phi_1\|^2_{H^1_g}+\|\phi_2\|^2_{H^1_g}\)^{p'} \( \eps^{2}+\|\phi_1\| _{H^1_g}+\|\phi_2\| _{H^1_g}\)^{p'} \|\phi_2\| ^{p'}_{H^1_g}\|\phi_1-\phi_2 \|^{p'}_{H^1_g}
 \\
&  \(\hbox{because  $\|\phi_i\| _{H^1_g}\le \sqrt\eps\|\phi_i\| _{\eps}\le c\eps^{5/2} $}\)\\
&\le   c\eps^{{4}p'-3}\|\phi_1-\phi_2\|_\eps
\end{align*}
Collecting the estimates of $I_i$'s we get \eqref{se12}.

\end{proof}

\begin{proof}[Proof of Proposition \ref{resto} (completed)]
By Proposition \ref{inv}, we deduce
$$\left\|T_{\eps,\xi}(\phi)\right\|_\eps\le c\left(\left\|N_{\eps,\xi}(\phi)\right\|_\eps+\left\|S_{\eps,\xi}(\phi)\right\|_\eps+
\left\|R_{\eps,\xi} \right\|_\eps\right)$$
and
$$\left\|T_{\eps,\xi}(\phi_1)-T_{\eps,\xi}(\phi_2)\right\|_\eps\le c \left\|N_{\eps,\xi}(\phi_1)-N_{\eps,\xi}(\phi_2)
\right\|_\eps+c\left\|S_{\eps,\xi}(\phi_1)-S_{\eps,\xi}(\phi_2)
\right\|_\eps.$$
By Lemma \ref{ne1} and Lemma \ref{se1} together with Proposition \ref{re1}, we immediately deduce that
$T_{\eps,\xi}$ is a contraction in the ball
centered at $0$ with radius $c\eps^2$
in $K_{\eps,\xi}^\perp$ for a suitable constant $c.$ Then $T_{\eps,\xi}$ has a unique fixed point.

In order to prove that the map $\xi\rightarrow \phi_{\eps,\xi}$ is a $C^1-$map, we apply the Implicit Function Theorem
to the $C^1-$function $G:M\times H_\eps\rightarrow H_\eps$ defined by
$$G( \xi,u):=\Pi^\perp_{\eps,\xi}\left\{W_{\eps,\xi}+\Pi^\perp_{\eps,\xi}u-i^*_\eps
\left[f\left(W_{\eps,\xi}+\Pi^\perp_{\eps,\xi}u\right)+\omega^2g\left(W_{\eps,\xi}+\Pi^\perp_{\eps,\xi}u\right)\right]\right\}+\Pi _{\eps,\xi}u.$$
 Indeed, $G\left(\xi,\phi_{\eps,\xi}\right)=0$ and the linearized operator
  ${\partial G\over\partial u}\left(\xi,\phi_{\eps,\xi}\right):H_\eps\rightarrow H_\eps$ defined by
$${\partial G\over\partial u}\left(\xi,\phi_{\eps,\xi}\right)(u)=
\Pi^\perp_{\eps,\xi}\left\{ \Pi^\perp_{\eps,\xi}(u)-i^*_\eps
\left[f'\left(W_{\eps,\xi}+\phi_{\eps,\xi}\right) \Pi^\perp_{\eps,\xi}(u)+\omega^2
g'\left(W_{\eps,\xi}+\phi_{\eps,\xi}\right) \Pi^\perp_{\eps,\xi}(u)\right]\right\}+\Pi _{\eps,\xi}(u)$$
is invertible, provided $\eps$ is small enough.
For any $\phi$ with $\left\|\phi \right\|_\eps\le c\eps^2$ it holds true that
\begin{eqnarray*}
& &\left\|{\partial G\over\partial u}\left(\xi,\phi_{\eps,\xi}\right)(u)\right\|_\eps
\\ & & \ge c\left\|\Pi _{\eps,\xi}(u)\right\|_\eps+
c\left\|\Pi^\perp_{\eps,\xi}\left\{ \Pi^\perp_{\eps,\xi}(u)-i^*_\eps
\left[f'\left(W_{\eps,\xi}+\phi_{\eps,\xi}\right) \Pi^\perp_{\eps,\xi}(u)+\omega^2
g'\left(W_{\eps,\xi}+\phi_{\eps,\xi}\right) \Pi^\perp_{\eps,\xi}(u)\right]\right\}\right\|_\eps
\\ & & \ge c\left\|\Pi _{\eps,\xi}(u)\right\|_\eps+
c\left\|L_{\eps,\xi}\left( \Pi^\perp_{\eps,\xi}(u)\right)\right\|_\eps
\\ & &- c\left\|\Pi^\perp_{\eps,\xi}\left\{ i^*_\eps
\left[\left(f'\left(W_{\eps,\xi}+\phi_{\eps,\xi}\right)-f'\left(W_{\eps,\xi}\right)\right)  \Pi^\perp_{\eps,\xi}(u)\right]\right\}\right\|_\eps
\\ & &- c\left\|\Pi^\perp_{\eps,\xi}\left\{ i^*_\eps
\left[\omega^2 g'\left(W_{\eps,\xi}+\phi_{\eps,\xi}\right)  \Pi^\perp_{\eps,\xi}(u)\right]\right\}\right\|_\eps
\\ & & \ge c\left\|\Pi _{\eps,\xi}(u)\right\|_\eps+
c\left\|  \Pi^\perp_{\eps,\xi}(u) \right\|_\eps
- c\eps^{2\min\{p-2,1\}}\left\|  \Pi^\perp_{\eps,\xi}(u) \right\|_\eps
-c\eps^{{3}-{3\over p'}}\left\|  \Pi^\perp_{\eps,\xi}(u) \right\|_\eps\\ & & \ge c \left\|u \right\|_\eps.\end{eqnarray*}
Indeed, at page 246 of \cite{MP} we proved that
\begin{eqnarray*}
& &\left\|\Pi^\perp_{\eps,\xi}\left\{ i^*_\eps
\left[\left(f'\left(W_{\eps,\xi}+\phi_{\eps,\xi}\right)-f'\left(W_{\eps,\xi}\right)\right) \Pi^\perp_{\eps,\xi}(u)\right]\right\}\right\|_\eps
\\ & &\le c\left( \left  \|\phi\right\|^{p-2} _{  \eps}+\left  \|\phi\right\| _{  \eps}\right)\left\|  \Pi^\perp_{\eps,\xi}(u) \right\|_\eps.
\end{eqnarray*}

Moreover we have

\begin{eqnarray*}
& &\left\|\Pi^\perp_{\eps,\xi}\left\{ i^*_\eps
\left[\omega^2 g'\left(W_{\eps,\xi}+\phi_{\eps,\xi}\right)  \Pi^\perp_{\eps,\xi}(u)\right]\right\}\right\|_\eps
\\ & &\le c\left| \(W_{\eps,\xi}+\phi_{\eps,\xi}\)
\left(2q-2q^2\Psi \(W_{\eps,\xi}+\phi_{\eps,\xi}\)\right)\Psi' \(W_{\eps,\xi}+\phi_{\eps,\xi}\) \left[\Pi^\perp_{\eps,\xi}(u)\right]\right |_{\eps,p'}
\\ & &+ c\left| \left[2q\Psi\(W_{\eps,\xi}+\phi_{\eps,\xi}\)- q^2\Psi ^2\(W_{\eps,\xi}+\phi_{\eps,\xi}\)\right]  \Pi^\perp_{\eps,\xi}(u) \right |_{\eps,p'}
\\ & &=I_1+I_2,
\end{eqnarray*}
  by Lemma \ref{lem:e7} we get
\begin{eqnarray*}
& &I_1\le{1\over \eps^{3\over p'}}
\left | W_{\eps,\xi}+\phi_{\eps,\xi} \right |_{g,6}\left | \Psi' \(W_{\eps,\xi}+\phi_{\eps,\xi}\)\Pi^\perp_{\eps,\xi}(u)\right |_{g,6}
\left |2q-2q^2\Psi \(W_{\eps,\xi}+\phi_{\eps,\xi}\)\right |_{g,{3-p'\over3p'}}\\
& &\le c{1\over \eps^{3\over p'}}\(\sqrt\eps\( {1\over\eps^3}\int_MW_{\eps,\xi} ^6\)^{1\over6}+\|\phi\|_{H^1_g} \)\(\eps^2+\|\phi\|_{H^1_g} \)
\left \|\Pi^\perp_{\eps,\xi} u\right \|_{H^1_g}\\ & &\le  c{1\over \eps^{3\over p'}}  \eps^3
\left \|\Pi^\perp_{\eps,\xi} u\right \|_{\eps}
 \end{eqnarray*}
and by Lemma \ref{lem:e3} we get
\begin{eqnarray*}
& &I_2\le{1\over \eps^{3\over p'}}
\left |\Pi^\perp_{\eps,\xi} u\right |_{g,6}
\left |\Psi \(W_{\eps,\xi}+\phi_{\eps,\xi}\)\right |_{g,6}
\left |2q-q^2\Psi \(W_{\eps,\xi}+\phi_{\eps,\xi}\)\right |_{g,{3-p'\over3p'}}\\ & &\le c{1\over \eps^{3\over p'}}\(\eps^{5\over2}+\|\phi\|_\eps^2\)
\left \|\Pi^\perp_{\eps,\xi} u\right \|_{H^1_g}\le  c{1\over \eps^{3\over p'}} \eps^{{5\over2}}\sqrt\eps
\left \|\Pi^\perp_{\eps,\xi} u\right \|_{\eps}.
 \end{eqnarray*}

That concludes the proof.\end{proof}

\section{The reduced energy}\label{red-exp}

This section is devoted to the proof of Proposition \ref{espa}.

The first important result is the following one.

\begin{lem}\label{fine4}
It holds true that
\begin{equation}\label{fine41}
 \widetilde I_\eps(\xi)= I_\eps\left(W_{\eps,\xi}  +\phi_{\eps,\xi} \right)=I_\eps\left(W_{\eps,\xi}
 \right)+o\left(\eps^2\right)=
J_\eps\left(W_{\eps,\xi}
 \right)+\frac{\omega^2}{2}G_\eps\left(W_{\eps,\xi}
 \right)+o\left(\eps^2\right)
\end{equation}
uniformly with respect to
$\xi $  as $\eps$ goes to zero.

Moreover, setting $ \xi(y)={\rm exp}_\xi(y),$ $y\in B(0,r)$ it holds true that
\begin{eqnarray}\label{fine42}
& & \left({\partial\over\partial y_h} \widetilde I_\eps(\xi(y))
\right)_{|_{y=0}}=\left({\partial\over\partial y_h}I_\eps\left(W_{\eps,\xi(y)}  +\phi_{\eps,\xi(y)} \right)
\right)_{|_{y=0}}\nonumber\\ & &=\left({\partial\over\partial y_h}I_\eps\left(W_{\eps,\xi(y)}  \right)\right)_{|_{y=0}}+o\left(\eps^2\right)
\nonumber\\ & &=\left({\partial\over\partial y_h}J_\eps\left(W_{\eps,\xi(y)}  \right)\right)_{|_{y=0}}+
\frac{\omega^2}{2}\left({\partial\over\partial y_h}G_\eps\left(W_{\eps,\xi(y)}  \right)\right)_{|_{y=0}}+o\left(\eps^2\right)
\end{eqnarray}
uniformly with respect to
$\xi $  as $\eps$ goes to zero.
 \end{lem}
\begin{proof}
We argue exactly as in Lemma 5.1 of \cite{MP}, once we prove the
the following estimates:
\begin{eqnarray}\label{new1}
& &G_\eps\left(W_{\eps,\xi}  +\phi_{\eps,\xi} \right)-G_\eps\left(W_{\eps,\xi}   \right)=o(\eps^2)
\end{eqnarray}
 \begin{eqnarray}\label{new2}
& &\left[G'_\eps\left(W_{\eps,\xi_0}  +\phi_{\eps, \xi_0}\right)-G'_\eps\left(W_{\eps,\xi_0}  \right)\right]
\left[\left({\partial\over\partial y_h} W_{\eps,\xi(y)}\right)_{|_{y=0}}\right]
=o(\eps^2) \end{eqnarray}
and
\begin{eqnarray}\label{new3}
& & G'_\eps\left(W_{\eps,\xi(y)}  +\phi_{\eps, \xi(y)}\right)
\left[ {\partial\over\partial y_h} \phi_{\eps,\xi(y)} \right] =o(\eps^2)
\end{eqnarray}

 Let us prove (\ref{new1}). We have (for some $\theta\in[0,1]$)
 \begin{eqnarray*}
& &G_\eps\left(W_{\eps,\xi}  +\phi_{\eps,\xi} \right)-G_\eps\left(W_{\eps,\xi}  \right)\\ & &={1\over\eps^3}\int\limits_M\left[
\Psi\left(W_{\eps,\xi}  +\phi_{\eps,\xi} \right)\left(W_{\eps,\xi}  +\phi_{\eps,\xi} \right)^2-\Psi\left(W_{\eps,\xi}  \right)\left(W_{\eps,\xi}   \right)^2\right]
\\ & &
={1\over\eps^3}\int\limits_M
\Psi'\left(W_{\eps,\xi} +\theta\phi_{\eps,\xi} \right)[\phi_{\eps,\xi}] \left(W_{\eps,\xi}   \right)^2+{1\over\eps^3}\int\limits_M \Psi\left(W_{\eps,\xi}  +\phi_{\eps,\xi} \right)
\left[2
 \phi_{\eps,\xi}   W_{\eps,\xi} +\phi^2 _{\eps,\xi}  \right]
\\ & &=:I_1+I_2
\end{eqnarray*}
with
\begin{eqnarray*}
& &I_1\le {1\over\eps^3}\(\int\limits_M \left[\Psi'\left(W_{\eps,\xi} +\theta\phi_{\eps,\xi} \)\right]^2\)^{1/2}\( {1\over\eps^3}\int\limits_M W_{\eps,\xi} ^4\)^{1/2}\eps^{3/2}
\\ & &\le c{1\over\eps^{3/2}}\left\|\Psi'\left(W_{\eps,\xi} +\theta\phi_{\eps,\xi} \)[\phi_{\eps,\xi}]\right\|_{H^1_g}
\quad\(\hbox{because of Lemma \ref{lem:e7}}\)\\
& &\le c{1\over\eps^{3/2}}\(\eps^2\|\phi\|_{H^1_g}+\|\phi\|_{H^1_g}^2\)
\quad \(\hbox{because  $\|\phi_{\eps,\xi}  \| _{H^1_g}\le \sqrt\eps\|\phi_{\eps,\xi} \| _{\eps}\le c\eps^{5/2} $}\)\\
& &=o(\eps^2)
\end{eqnarray*}
and
\begin{eqnarray*}
& &I_2\le {1\over\eps^3}\(\int\limits_M \left(\Psi \left(W_{\eps,\xi} + \phi_{\eps,\xi} \)[\phi_{\eps,\xi}]\right)^2\)^{1/2}\(  \int\limits_M \phi_{\eps,\xi} ^4\)^{1/2} \\ & &+
{1\over\eps^3}\(\int\limits_M \left[\Psi \left(W_{\eps,\xi} + \phi_{\eps,\xi} \)\right]^3\)^{1/3}\(  \int\limits_M \phi_{\eps,\xi} ^3\)^{1/3} \( {1\over\eps^3}\int\limits_M W_{\eps,\xi} ^3\)^{1/3}\eps
\\ & &\le c{1\over\eps^{3 }}\left\|\Psi \left(W_{\eps,\xi} + \phi_{\eps,\xi} \)\right\|_{H^1_g} \left\|  \phi_{\eps,\xi} \right\|^2_{H^1_g} +
c{1\over\eps^{2 }}\left\|\Psi \left(W_{\eps,\xi} + \phi_{\eps,\xi} \)\right\|_{H^1_g} \left\|  \phi_{\eps,\xi} \right\|_{H^1_g}
\\ & &\(\hbox{because of Lemma \ref{lem:e3} and the fact that $\|\phi_{\eps,\xi}  \| _{H^1_g}
\le \sqrt\eps\|\phi_{\eps,\xi} \| _{\eps}\le c\eps^{5/2} $}\)\\
& &\le c  \eps^{5+{5\over2}-3} +c \eps^{ {5\over2}+{5\over2}-2}=o(\eps^2).
\end{eqnarray*}
Then (\ref{new1}) follows.

 Let us prove (\ref{new2}). We have (for some $\theta\in[0,1]$)
\begin{eqnarray*}
& &\left[G'_\eps\left(W_{\eps,\xi_0}  +\phi_{\eps, \xi_0}\right)-G'_\eps\left(W_{\eps,\xi_0}  \right)\right]
\left[\left({\partial\over\partial y_h} W_{\eps,\xi(y)}\right)_{|_{y=0}}\right]
\\ & &\le
\left| {q\over2\eps^3 }\int\limits_M\left\{ \left[2\Psi \left(W_{\eps,\xi_0} + \phi_{\eps,\xi_0} \)-\Psi \left(W_{\eps,\xi_0}  \)\right]-\left[q\Psi ^2\left(W_{\eps,\xi_0} + \phi_{\eps,\xi_0} \)-q\Psi ^2\left(W_{\eps,\xi_0}  \)\right]\right\}W_{\eps,\xi_0} \left({\partial\over\partial y_h} W_{\eps,\xi(y)}\right)_{|_{y=0}}\right|\\ & &+
\left| {q\over2\eps^3 }\int\limits_M  \left[2\Psi \left(W_{\eps,\xi_0} + \phi_{\eps,\xi_0} \) - q\Psi ^2\left(W_{\eps,\xi_0} + \phi_{\eps,\xi_0} \) \right]\phi_{\eps,\xi_0} \left({\partial\over\partial y_h} W_{\eps,\xi(y)}\right)_{|_{y=0}}\right|
\\ & &\le
\left| {q\over2\eps^3 }\int\limits_M  2\Psi' \left(W_{\eps,\xi_0} +\theta \phi_{\eps,\xi_0} \)\(\phi_{\eps,\xi_0}\)   W_{\eps,\xi_0} \left({\partial\over\partial y_h} W_{\eps,\xi(y)}\right)_{|_{y=0}}\right|\\
 & &+
\left| {q^2\over\eps^3 }\int\limits_M  \Psi \left(W_{\eps,\xi_0} +\theta \phi_{\eps,\xi_0} \)\Psi' \left(W_{\eps,\xi_0} +\theta \phi_{\eps,\xi_0} \)\(\phi_{\eps,\xi_0}\)   W_{\eps,\xi_0} \left({\partial\over\partial y_h} W_{\eps,\xi(y)}\right)_{|_{y=0}}\right|\\
& &+ \left|{q \over  \eps^3 }\int\limits_M  \Psi  \left(W_{\eps,\xi_0}  \)       \phi_{\eps,\xi_0}  \left({\partial\over\partial y_h} W_{\eps,\xi(y)}\right)_{|_{y=0}}\right|\\
\\
& &+ \left|{q \over2 \eps^3 }\int\limits_M  \Psi '\left(W_{\eps,\xi_0}  \) \(\phi_{\eps,\xi_0}\)     \phi_{\eps,\xi_0}  \left({\partial\over\partial y_h} W_{\eps,\xi(y)}\right)_{|_{y=0}}\right|\\
  & &+\left| {q^2 \over2 \eps^3 }\int\limits_M  \Psi ^2\left(W_{\eps,\xi_0}  \)       \phi_{\eps,\xi_0}  \left({\partial\over\partial y_h} W_{\eps,\xi(y)}\right)_{|_{y=0}}\right|\\
  & &=:I_1+I_2+I_3+I_4+I_5.
  \end{eqnarray*}

By Lemma \ref{lem:e7} and the facts that $\left \| \phi_{\eps,\xi(y) } \right \|_{H^1_g}=O\(\eps^{5/2}\)$ and
Remark \ref{remark:Weps} we get
\begin{eqnarray*}
& & I_1\le c{1\over \eps  }\(\int\limits_M \left[\Psi' \left(W_{\eps,\xi_0} +\theta \phi_{\eps,\xi_0} \)\(\phi_{\eps,\xi_0}\) \right]^3\)^{1/3}
\({1\over \eps^3 }\int\limits_MW_{\eps,\xi_0}^3\)^{1/3}  \({1\over \eps^3 }\int\limits_M\left[ \left({\partial\over\partial y_h} W_{\eps,\xi(y)}\right)_{|_{y=0}}\right]^3\)^{1/3}  \\
& & \(\hbox{we use (6.3) of \cite{MP}}\)\\
& &\le c \eps^{7/2} \(\int\limits_M\left[  \sum\limits_{k=1}^3\left|{1\over \eps}{\partial U(z)\over\partial z_k} \chi(\eps z)+{\partial \chi(\eps z)\over\partial z_k} U(  z)
\right|\left| \delta_{hk}+\eps^2|z|^2
\right|\right]^3\)^{1/3 }=O\(\eps^{5/2}\)=o(\eps^2). \end{eqnarray*}
By the estimate of $I_1$ we get $I_2=o(\eps^2),$ because of Lemma \ref{lem:e3} and we also get $I_4=o(\eps^2)$,because $\|\phi_{\eps,\xi_0}\|_{H^1_g}=O\(\eps^{5/2}\)$.
Let us estimate $I_3.$ We use the definition of $\widetilde v_\eps$ given in Lemma \ref{lem:e5} and we get
\begin{eqnarray*}
& & I_3=\left|{1\over \eps ^3 } \int\limits_{B_g(\xi_0,R)} \Psi \left(W_{\eps,\xi_0} \)\(\phi_{\eps,\xi_0}\)\left({\partial\over\partial y_h} W_{\eps,\xi(y)}\right)_{|_{y=0}}\right| \\
& &\le   c  \int\limits_{B ( 0,R/\eps)}|\widetilde v_\eps(z)|\left|\widetilde \phi_{\eps,\xi_0}(z)\right|  |  \sum\limits_{k=1}^3\left|{1\over \eps}{\partial U(z)\over\partial z_k} \chi(\eps z)+{\partial \chi(\eps z)\over\partial z_k} U(  z)
\right|\left| \delta_{hk}+\eps^2|z|^2
\right|dz \\
 & &\le   c  \(\int\limits_{B ( 0,R/\eps)}\widetilde \phi_{\eps,\xi_0}^2(z)\)^{1/2}\(\eps^4
 \int\limits_{B ( 0,R/\eps)}{\widetilde v_\eps^2(z) \over\eps^4}  \left[  \sum\limits_{k=1}^3\left|{1\over \eps}{\partial U(z)\over\partial z_k} \chi(\eps z)+{\partial \chi(\eps z)\over\partial z_k} U(  z)
\right|\left| \delta_{hk}+\eps^2|z|^2
\right|\right]^2dz\)^{1/2}\\ & &\(\hbox{because ${\widetilde v_\eps^2(z) \over\eps^4}\to \gamma$ weakly in $L^6$ }\)\\ & &=o(\eps^2). \\
 \end{eqnarray*}
Here we used the fact that (see (6.3) of \cite{MP}) the function $\widetilde \phi_{\eps,\xi_0} (z):=  \phi_{\eps,\xi_0}\({\rm exp}_{\xi_0}(\eps z)\)= \phi_{\eps,\xi_0}(x)$
can be estimated as 
 \begin{eqnarray*}& &\|\phi_{\eps,\xi_0}\|^2_\eps\ge {1\over \eps ^3 } \int\limits_{B_g(\xi_0,R)}\(\eps^2|\nabla_g\phi_{\eps,\xi_0}|^2+\phi_{\eps,\xi_0}^2\)d\mu_g=
  \\ & &={1\over \eps ^3 } \int\limits_{B ( 0,R/\eps)}
 \( \sum\limits_{i,j=1}^3\eps^2g^{ij}(\eps z){\partial \widetilde \phi_{\eps,\xi_0}\over\partial z_i}{\partial \widetilde \phi_{\eps,\xi_0}\over\partial z_j}+\widetilde \phi_{\eps,\xi_0}^2\)|g(\eps z)|^{1/2}dz\\ & &\ge c\int\limits_{B ( 0,R/\eps)} \widetilde \phi_{\eps,\xi_0}^2(z)dz
 \end{eqnarray*}
which implies
$$\int\limits_{B ( 0,R/\eps)} \widetilde \phi_{\eps,\xi_0}^2(z)dz=O\(\|\phi_{\eps,\xi_0}\|^2_\eps\)=O\(\eps^4\).$$

By the estimate of $I_3$ we get $I_5=o(\eps^2),$ because of Lemma \ref{lem:e3}.

 Let us prove (\ref{new3}). We have (for some $\theta\in[0,1]$)
Arguing as in the proof of (5.10) of \cite{MP}, the proof of \eqref{new3} reduces
to the proof of the following estimate
\begin{eqnarray}\label{new4}
& &{1\over\eps^3} \int\limits_Mg\(W_{\eps,\xi (y)} +  \phi_{\eps,\xi(y) }\)\(W_{\eps,\xi(y) } +  \phi_{\eps,\xi(y) }\)Z^l_{\eps,\xi(y)} =o(\eps^2),\end{eqnarray}
where the functions $Z^l_{\eps,\xi(y)} $ are defined in \eqref{zi}.
First of all we point out that
\begin{eqnarray}\label{new5}
& &{1\over\eps^3}\int\limits_M   \Psi \(W_{\eps,\xi (y)} \)  W_{\eps,\xi(y) } Z^l_{\eps,\xi(y)}  =o(\eps^2).\end{eqnarray}
By Lemma \ref{lem:e5} we have that ${\displaystyle \left\{ \frac{1}{\varepsilon^{2}}\tilde{v}_{\varepsilon,\xi}\right\} _{n}}$
converges to $\gamma$ weakly in $L^{6}(\mathbb{R}^{3})$. So, arguing as in the proof of \eqref{eq13} and using Lemma \ref{lem:e5} we get
\begin{eqnarray*}
& &{1\over\eps^3}\int\limits_M   \Psi \(W_{\eps,\xi (y)} \)  W_{\eps,\xi(y) } Z^l_{\eps,\xi(y)}={1\over\eps^3}\int\limits_{B_g(\xi,R)}   \Psi \(W_{\eps,\xi (y)} \)  W_{\eps,\xi(y) } Z^l_{\eps,\xi(y)} \\ & &
=  \int\limits_{B(0,R/\eps)}   \tilde v(z) \chi^2(\eps z) U(z){\partial U\over\partial z_l}(z)dz=
{1\over2}\eps^2\int\limits_{B(0,R/\eps)} \frac{  \tilde v(z)}{\eps^2} \chi^2(\eps z)  {\partial  U^2\over\partial z_l} dz=
\\&&=\frac{1}{\eps^2}\int_{\mathbb{R}^3}\gamma  {\partial  U^2\over\partial z_l} dz+o(\eps^2)=o(\eps^2).
\end{eqnarray*}
because both $U$ and $\gamma$ are radially symmetric.

 Therefore, by \eqref{new5}, using the definition of $g$ in  \eqref{go}  we are lead to estimate
 (for some $\theta\in[0,1]$)
\begin{eqnarray*}
& &{1\over\eps^3}\left|\int\limits_Mg\(W_{\eps,\xi (y)} +  \phi_{\eps,\xi(y) }\)\(W_{\eps,\xi(y) } +  \phi_{\eps,\xi(y) }\)Z^l_{\eps,\xi(y)} +2q\int\limits_M   \Psi \(W_{\eps,\xi (y)} \)  W_{\eps,\xi(y) } Z^l_{\eps,\xi(y)}  \right|\nonumber\\
& &={1\over\eps^3}\left|\int\limits_M\left[q^2\Psi^2\(W_{\eps,\xi (y)} +  \phi_{\eps,\xi(y) }\)-2q \Psi \(W_{\eps,\xi (y)} +  \phi_{\eps,\xi(y) }\)\right]\(W_{\eps,\xi(y) } + \phi_{\eps,\xi(y) }\)Z^l_{\eps,\xi(y)} \right.\\ & &\left.
+2q\int\limits_M   \Psi \(W_{\eps,\xi (y)} \)  W_{\eps,\xi(y) } Z^l_{\eps,\xi(y)}  \right|\\
& &\le c{1\over\eps^3}\int\limits_M \left| \Psi ^2\(W_{\eps,\xi (y)} +  \phi_{\eps,\xi(y) }\) \(W_{\eps,\xi(y) } + \phi_{\eps,\xi(y) }\)Z^l_{\eps,\xi(y)} \right|\nonumber\\
& &+c{1\over\eps^3}\int\limits_M  \left|\Psi \(W_{\eps,\xi (y)} \)   \phi_{\eps,\xi(y) } Z^l_{\eps,\xi(y)} \right|\\ & &
+c{1\over\eps^3}\int\limits_M  \left|\Psi' \(W_{\eps,\xi (y)} + \theta \phi_{\eps,\xi(y) }\)\( \phi_{\eps,\xi(y) }\) \(W_{\eps,\xi(y) } + \phi_{\eps,\xi(y) }\)Z^l_{\eps,\xi(y)} \right|=I_1+I_2+I_3 ,\end{eqnarray*}
with

\begin{eqnarray*}
&&I_1\le {1\over\eps^3}
 \left |\Psi \(W_{\eps,\xi (y)} +   \phi_{\eps,\xi(y) } \)  \right |^2_{g,3}\(\left | W_{\eps,\xi(y) }\right |_{g,3} + \left | \phi_{\eps,\xi(y) } \right |_{g,3}\)
\left | Z^l_{\eps,\xi(y)} \right |_{g,3}\\ & &\(\hbox{we use Lemma \ref{lem:e3}}\)
\\ & &\le {1\over\eps^3}
 \left \|\Psi \(W_{\eps,\xi (y)} +   \phi_{\eps,\xi(y) } \)  \right \|^2_{H^1_g}\(\left | W_{\eps,\xi(y) }\right |_{g,3} + \left \| \phi_{\eps,\xi(y) } \right \|_{H^1_g}\)
\left | Z^l_{\eps,\xi(y)} \right |_{g,3}\\ & &\(\hbox{we use  that $\|\phi  \| _{H^1_g}\le \sqrt\eps\|\phi \| _{\eps}\le c\eps^{5/2} $, $\left | W_{\eps,\xi(y) }\right |_{g,3}=O(\eps)$
and  $\left | Z^l_{\eps,\xi(y)} \right |_{g,3}=O(\eps)$}\)\\
& &=o(\eps^2),\end{eqnarray*}
\begin{eqnarray*}
& &I_2\le {1\over\eps^3}
 \left\|\Psi \(W_{\eps,\xi (y)} \)  \right\|_{H^1_g}\left\| \phi_{\eps,\xi(y) }  \right\|_{H^1_g}
\left | Z^l_{\eps,\xi(y)} \right |_{g,3}\\ & &\(\hbox{we use Lemma \ref{lem:e3} and also that $\|\phi  \| _{H^1_g}\le \sqrt\eps\|\phi \| _{\eps}\le c\eps^{5/2} $ and  $\left | Z^l_{\eps,\xi(y)} \right |_{g,3}=O(\eps)$}\)\\
& &=o(\eps^2)\end{eqnarray*}
and
\begin{eqnarray*}
& &I_3 \le {1\over\eps^3}
 \left |\Psi' \(W_{\eps,\xi (y)} + \theta \phi_{\eps,\xi(y) }\)\( \phi_{\eps,\xi(y) }\) \right |_{g,3}\left | W_{\eps,\xi(y) } + \phi_{\eps,\xi(y) } \right |_{g,3}
\left | Z^l_{\eps,\xi(y)} \right |_{g,3}\\ & &\le {1\over\eps^3}
 \left \|\Psi' \(W_{\eps,\xi (y)} + \theta \phi_{\eps,\xi(y) }\)\( \phi_{\eps,\xi(y) }\) \right \|_{H^1_g}\(\left | W_{\eps,\xi(y) }\right |_{g,3} + \left \| \phi_{\eps,\xi(y) } \right \|_{H^1_g}\)
\left | Z^l_{\eps,\xi(y)} \right |_{g,3}\\ & &\(\hbox{because of Lemma \ref{lem:e7}}\)\\ & &\le {1\over\eps^3}
\(\eps^2+ \left \| \phi_{\eps,\xi(y) }  \right \|_{H^1_g}\)\left \| \phi_{\eps,\xi(y) }  \right \|_{H^1_g}\(\left | W_{\eps,\xi(y) }\right |_{g,3} + \left \| \phi_{\eps,\xi(y) } \right \|_{H^1_g}\)
\left | Z^l_{\eps,\xi(y)} \right |_{g,3}\\ & &\(\hbox{we use  that $\|\phi  \| _{H^1_g}\le \sqrt\eps\|\phi \| _{\eps}\le c\eps^{5/2} $, $\left | W_{\eps,\xi(y) }\right |_{g,3}=O(\eps)$
and  $\left | Z^l_{\eps,\xi(y)} \right |_{g,3}=O(\eps)$}\)\\
& &=o(\eps^2).\end{eqnarray*}

\end{proof}

\begin{lem}\label{l5}
It holds  true that
\begin{equation}\label{l53}
  J_\eps\left( W_{\eps,\xi}\right)
=C-\alpha{\eps^2\over6}  S_g(\xi) +o\left(\eps^2\right)
,\end{equation}
$C^1-$uniformly with respect to $\xi\in M$ as $\varepsilon$ goes to zero.
Here
$$C:={1\over2}\int\limits_{\rr^3}|\nabla U|^2dz-{\lambda\over 2}\int\limits_{\rr^3}U^2dz
-{1\over p}\int\limits_{\rr^3}U^pdz$$
and
$$\alpha:=\int\limits_{\rr^3}\left({  U'(|z|)\over|z|}\right)^2z_1^4dz.$$\end{lem}
 \begin{proof} See  Lemma (4.2) of \cite{MP}.\end{proof}

\begin{lem}\label{l5bis}
It holds true that
\[
G_{\varepsilon}(W_{\varepsilon,\xi}):=\frac{1}{\varepsilon^{3}}
\int_{B_{g}(\xi,r)}\Psi(W_{\varepsilon,\xi})W_{\varepsilon,\xi}^{2}d\mu_{g}=\beta\varepsilon^{2}+o(\varepsilon^{2})
\]
 $C^1-$uniformly with respect to $\xi\in M$ as $\varepsilon$ goes to zero.
Here
\[
\beta=\int_{\mathbb{R}^{3}}\gamma(z)U^{2}(z)dz=\frac 1q \int_{\mathbb{R}^{3}}|\nabla\gamma(z)|^{2}dz
\]
 with $\gamma\in D^{1,2}(\mathbb{R}^{3})$ such that $-\Delta\gamma=qU^{2}$.\end{lem}
\begin{proof} {\em Step 1:  the $C^0$-estimate.}

By the weak convergence of ${\displaystyle \left\{ \frac{1}{\varepsilon_{n}^{2}}v_{\varepsilon_{n},\xi}\right\} _{n}}$
in $L^{6}(\mathbb{R}^{3})$ we infer
\begin{eqnarray*}
\frac{G_{\varepsilon_{n}}(W_{\varepsilon_{n},\xi})}{\varepsilon_{n}^{2}} & =
& \int_{\mathbb{R}^{3}}\frac{\tilde{v}_{\varepsilon_{n},\xi}(z)}{\varepsilon_{n}^{2}}
\chi_{r}^{2}(\varepsilon_{n}|z|)U^{2}(z)|g_{\xi}(\varepsilon_{n}z)|^{1/2}\\
 & = & \int_{\mathbb{R}^{3}}\frac{v_{\varepsilon_{n},\xi}(z)}{\varepsilon_{n}^{2}}
 \chi_{r}(\varepsilon_{n}|z|)U^{2}(z)|g_{\xi}(\varepsilon_{n}z)|^{1/2}\rightarrow\int_{\mathbb{R}^{3}}\gamma(z)U^{2}(z)dz.
\end{eqnarray*}
We have to prove that the convergence is uniform with respect to $\xi\in M$.

By the expansions of $|g_\xi(\varepsilon z)|^{1/2}$ and $\chi(\varepsilon |z|)$, and by (\ref{eq:egamma-sol}) we have

\begin{eqnarray}
\nonumber
\frac{G_{\varepsilon}(W_{\varepsilon,\xi})}{\varepsilon^{2}}
& = & \int_{\mathbb{R}^{3}}\frac{\tilde{v}_{\varepsilon,\xi}(z)}{\varepsilon^{2}}
\chi_{r}^{2}(\varepsilon |z|)U^{2}(z)|g_{\xi}(\varepsilon z)|^{1/2}dz\\
\nonumber
 & = & \frac 1q \int_{\mathbb{R}^{3}}
 \frac{\tilde{v}_{\varepsilon,\xi}(z)}{\varepsilon^{2}}
\chi_{r}^{2}(\varepsilon |z|)qU^{2}(z)dz+O(\varepsilon^2)\\
\label{eq:Geps}
& = &- \frac 1q \int_{\mathbb{R}^{3}}
 \frac{\tilde{v}_{\varepsilon,\xi}(z)}{\varepsilon^{2}}
\chi_{r}^{2}(\varepsilon |z|)\Delta \gamma dz+O(\varepsilon^2)\\
\nonumber
& = &- \frac 1q \int_{\mathbb{R}^{3}}
 \left(\Delta\frac{\tilde{v}_{\varepsilon,\xi}(z)}{\varepsilon^{2}}\right)
\chi_{r}^{2}(\varepsilon |z|) \gamma dz+O(\varepsilon)\\
\nonumber
&=&- \frac 1q \int_{\mathbb{R}^{3}}
-\sum_{ij}\partial_{j}\left(|g_{\xi}(\varepsilon z)|^{1/2}g_{\xi}^{ij}(\varepsilon z)\partial_{i}
\frac{\tilde{v}_{\varepsilon,\xi}(z)}{\varepsilon^{2}}\right)
\chi_{r}^{2}(\varepsilon |z|) \gamma dz+O(\varepsilon)
\end{eqnarray}
uniformly with respect to $\xi$ as $\varepsilon$ goes to zero.

By (\ref{eq:e11-1}) and by the expansions of $|g_\xi(\varepsilon z)|^{1/2}$ and $\chi(\varepsilon |z|)$ we have

\begin{multline}
\label{eq:Geps2}
\left|\frac{G_{\varepsilon}(W_{\varepsilon,\xi})}{\varepsilon^{2}} -
\int_{\mathbb{R}^{3}} U^2\gamma dz
\right| \le O(\varepsilon)+
\left| \int_{\mathbb{R}^{3}} U^2(z)\gamma(z) [|g_{\xi}(\varepsilon z)|^{1/2}\chi_{r}^{2}(\varepsilon |z|)-1]dz
\right|\le \\
+\frac 1q \left|
\int_{\mathbb{R}^{3}}|g_{\xi}(\varepsilon z)|^{1/2}
\left[\chi_{r}^{2}(\varepsilon |z|)+q^{2}U^{2}(z)\chi_{r}^{4}(\varepsilon |z|)\right]\tilde{v}_{\varepsilon,\xi}(z)\gamma(z) dz
\right|\le \\
\le O(\varepsilon)+ \frac 1q\left|
\int_{\mathbb{R}^{3}}
\tilde{v}_{\varepsilon,\xi}(z)\gamma(z)
\right|\le O(\varepsilon)
\end{multline}
uniformly with respect to $\xi$ as $\varepsilon$ goes to zero.

\bigskip {\em Step 2:  the $C^1$-estimate.}

More precisely, if
  $\xi(y)=\exp_{\xi}(y)$ for $y\in B(0,r),$. we are going to prove that
\begin{equation}\label{eq13}
\left.\frac{\partial}{\partial y_{h}}G_{\varepsilon}(W_{\varepsilon,\xi(h)})\right|_{y=0}=o(\varepsilon^{2})\
\hbox{uniformly with respect to $\xi$ as $\varepsilon$ goes to $0$.}
\end{equation}
 We have that
\begin{eqnarray*}
\left.\frac{\partial}{\partial y_{h}}G_{\varepsilon}(W_{\varepsilon,\xi})\right|_{y=0} & = &
\left.\frac{\partial}{\partial y_{h}}\frac{1}{\varepsilon^{3}}\int_{M}\Psi(W_{\varepsilon,\xi(y)})
W_{\varepsilon,\xi(y)}^{2}\right|_{y=0}d\mu_{g}\\
 & = & \left.\frac{1}{\varepsilon^{3}}\int_{M}\Psi(W_{\varepsilon,\xi(y)})
 2W_{\varepsilon,\xi(h)}\left(\frac{\partial}{\partial y_{h}}W_{\varepsilon,\xi(h)}\right)\right|_{y=0}d\mu_{g}\\
 &  & +\frac{1}{\varepsilon^{3}}\int_{M}W_{\varepsilon,\xi(h)}^{2}\Psi'(W_{\varepsilon,\xi(y)})
 \left[\left.\frac{\partial}{\partial y_{h}}W_{\varepsilon,\xi(h)}\right|_{y=0}\right]d\mu_{g}
\end{eqnarray*}
We call $I_{1}(\varepsilon,\xi)$ and $I_{2}(\varepsilon,\xi)$ respectively
the first and the second addendum of the above equation.

We recall that (see Section 6 of \cite{MP}) that
\begin{eqnarray}
&&\left.\frac{\partial}{\partial y_{h}}W_{\varepsilon,\xi(h)}\right|_{y=0}  \!\!\!
=
\sum_{k=1}^{3}\left[\frac{1}{\varepsilon}\frac{\partial U(z)}{\partial z_{k}}\chi_{r}(\varepsilon|z|)+
U(z)\frac{\partial\chi_{r}(\varepsilon|z|)}{\partial z_{k}}\right]\label{mp1}
\frac{\partial}{\partial y_{h}}\mathcal{E}_{k}(0,\exp_{\xi}(\varepsilon z))\\
&&\frac{\partial}{\partial y_{h}}\mathcal{E}_{k}(0,\exp_{\xi}(\varepsilon z))  =
\delta_{hk}+O(\varepsilon^{2}|z|^{2})\label{mp2}\\
&&|g_\xi(\varepsilon z)|^{1/2}  = 1-\frac{\varepsilon^{2}}{4}\sum_{i,s,k=1}^{3}
\frac{\partial^{2}g_\xi^{ij}(0)}{\partial z_{s}\partial z_{k}}z_{s}z_{k}+O(\varepsilon^{3}|z|^3)\label{mp3}
\end{eqnarray}
 Using the normal coordinates and the previous estimates we get
\begin{multline*}
\frac{1}{\varepsilon^{2}}I_{1}(\varepsilon,\xi)
=\int_{\mathbb{R}^3}\frac{\tilde{v}_{\varepsilon,\xi}(z)}{\varepsilon^{2}}2U(z)\chi_{r}(\varepsilon|z|)
|g_\xi(\varepsilon z)|^{1/2}\times \\
\times
\left\{ \sum_{k=1}^{3}
\left[
\frac{1}{\varepsilon}\frac{\partial U(z)}{\partial z_{k}}\chi_{r}(\varepsilon|z|)+
U(z)\frac{\partial\chi_{r}(\varepsilon|z|)}{\partial z_{k}}
\right]
\frac{\partial}{\partial y_{h}}\mathcal{E}_{k}(0,\exp_{\xi}(\varepsilon z))\right\}dz.
\end{multline*}
By Lemma \ref{lem:e5} we have that ${\displaystyle \left\{ \frac{1}{\varepsilon^{2}}\tilde{v}_{\varepsilon,\xi}\right\} _{n}}$
converges to $\gamma$ weakly in $L^{6}(\mathbb{R}^{3})$, so we have
\[
I_{1}(\varepsilon,\xi)=2\varepsilon\int_{\mathbb{R}^{3}}\gamma U(z)U'(z)\frac{z_{h}}{|z|}dz+o(\varepsilon^{2}).
\]
Finally, we have that ${\displaystyle \int_{\mathbb{R}^{3}}\gamma(z)U(z)U'(z)\frac{z_{h}}{|z|}dz=0}$
because both $\gamma$ and $U$ are radially symmetric on $z$.

At this point we have to prove the uniform convergence of $I_{1}(\varepsilon,\xi)$ with respect to $\xi\in M$.
We remark that, by (\ref{eq:egamma-sol}) we have, for all $k=1,2,3$,
$\displaystyle-\Delta \frac{\partial}{\partial z_k}\gamma(z)=\frac{\partial}{\partial z_k}U^2(z)$. Thus, by
(\ref{mp1}), (\ref{mp2}), (\ref{mp3}), we get

\begin{eqnarray*}
\frac{1}{\varepsilon^{2}}I_{1}(\varepsilon,\xi)
&=&\frac1q \int_{\mathbb{R}^3}\frac{\tilde{v}_{\varepsilon,\xi}(z)}{\varepsilon^{2}}\chi_{r}(\varepsilon|z|)
\frac{q}{\varepsilon}\frac{\partial U^2(z)}{\partial z_{k}}dz+O(\varepsilon)\\
&=&-\frac1{\varepsilon q} \int_{\mathbb{R}^3}\frac{\tilde{v}_{\varepsilon,\xi}(z)}{\varepsilon^{2}}\chi_{r}(\varepsilon|z|)
\Delta\left( \frac{\partial \gamma(z)}{\partial z_{k}}\right)dz+O(\varepsilon)\\
&=&-\frac1{\varepsilon q} \int_{\mathbb{R}^3}\Delta\left(\frac{\tilde{v}_{\varepsilon,\xi}(z)}{\varepsilon^{2}}\right)
\chi_{r}(\varepsilon|z|)
 \frac{\partial \gamma(z)}{\partial z_{k}}dz\\
 &&+\frac1{q} \int_{\mathbb{R}^3}\nabla\left(\frac{\tilde{v}_{\varepsilon,\xi}(z)}{\varepsilon^{2}}\right)
\chi'_{r}(\varepsilon|z|)\frac{z}{|z|}
 \frac{\partial \gamma(z)}{\partial z_{k}}dz
 +O(\varepsilon)
\end{eqnarray*}
Now we have that
\begin{multline*}
\frac1{q} \left|\int_{\mathbb{R}^3}\nabla\left(\frac{\tilde{v}_{\varepsilon,\xi}(z)}{\varepsilon^{2}}\right)
\chi'_{r}(\varepsilon|z|)\frac{z}{|z|}
\frac{\partial \gamma(z)}{\partial z_{k}}dz\right|
\le \\
\left\|\frac{\tilde{v}_{\varepsilon,\xi}(z)}{\varepsilon^{2}}
\right\|_{D^{1,2}(B(0,r/\varepsilon))}
\left\| \frac{\partial \gamma(z)}{\partial z_{k}}dz
\right\|_{D^{1,2}(\mathbb{R}^3\smallsetminus B(0,r/2\varepsilon))}
\end{multline*}
and the last term vanish uniformly in $\xi$ when $\varepsilon$ goes to zero because
$\displaystyle \frac{\partial \gamma(z)}{\partial z_{k}}$ decays exponentially with respect to $|z|$.

Moreover, arguing as in (\ref{eq:Geps}) and in (\ref{eq:Geps2}) we obtain
\begin{multline*}
-\frac1{\varepsilon q} \int_{\mathbb{R}^3}\Delta\left(\frac{\tilde{v}_{\varepsilon,\xi}(z)}{\varepsilon^{2}}\right)
\chi_{r}(\varepsilon|z|)
 \frac{\partial \gamma(z)}{\partial z_{k}}dz\\
= \frac1{\varepsilon} \int_{\mathbb{R}^3}U^2(z)
 \frac{\partial \gamma(z)}{\partial z_{k}}dz+O(\varepsilon)
= -\frac1{\varepsilon} \int_{\mathbb{R}^3} \frac{\partial }{\partial z_{k}}\left[U^2(z)\right]
\gamma(z)dz+O(\varepsilon)
\end{multline*}
and the last integral is zero because both $U,\gamma$ and $\chi_r$ are radially symmetric.

By Equation (\ref{eq:e1}), Lemma \ref{lem:e1}, and by (\ref{mp1}), (\ref{mp2}), (\ref{mp3}), we have
\begin{eqnarray*}
I_{2}(\varepsilon,\xi) & = &
\frac{1}{q\varepsilon^{3}}\int_{M}
\left\{-\Delta_{g}\Psi(W_{\varepsilon,\xi})+
(1+qW_{\varepsilon,\xi}^{2})\Psi(W_{\varepsilon,\xi})\right\}
\Psi'(W_{\varepsilon,\xi(y)})
\left[\left.\frac{\partial}{\partial y_{h}}W_{\varepsilon,\xi(h)}\right|_{y=0}\right]d\mu_{g}\\
 & = & \frac{1}{q\varepsilon^{3}}\int_{M}\Psi(W_{\varepsilon,\xi})
 \left\{-\Delta_{g}\Psi'(W_{\varepsilon,\xi})[\cdot]
 +(1+qW_{\varepsilon,\xi}^{2})\Psi'(W_{\varepsilon,\xi(y)})\left[\cdot\right]\right\}
 d\mu_{g}\\
 & = & \frac{1}{q\varepsilon^{3}}\int_{M}\Psi(W_{\varepsilon,\xi})
 \left\{2qW_{\varepsilon,\xi}(1-q\Psi(W_{\varepsilon,\xi}))
 \left[\left.\frac{\partial}{\partial y_{h}}W_{\varepsilon,\xi(h)}\right|_{y=0}\right]\right\}
 d\mu_{g}=\\
 & = & 2\int_{\mathbb{R}^{3}}|g(\varepsilon z)|^{1/2}\tilde{v}_{\varepsilon,\xi}(z)U(z)\chi_{r}(\varepsilon|z|)
 (1-q\tilde{v}_{\varepsilon,\xi}(z))\times\\
 &  & \times\left\{ \sum_{k=1}^{3}
 \left[
 \frac{1}{\varepsilon}\frac{\partial U(z)}{\partial z_{k}}\chi_{r}(\varepsilon|z|)+
 U(z)\frac{\partial\chi_{r}(\varepsilon|z|)}{\partial z_{k}}
 \right]
 \frac{\partial}{\partial y_{h}}\mathcal{E}_{k}(0,\exp_{\xi}(\varepsilon z))\right\} dz.
\end{eqnarray*}
Again, by Lemma \ref{lem:e5} and by (\ref{mp1}), (\ref{mp2}), (\ref{mp3}), we have that
\[
I_{2}(\varepsilon,\xi)=2\varepsilon\int_{\mathbb{R}^{3}}\gamma U(z)U'(z)\frac{z_{h}}{|z|}dz+o(\varepsilon^{2})
=o(\varepsilon^{2}).
\]

At this point we have to prove the uniform convergence of $I_{2}(\varepsilon,\xi)$ with respect to $\xi\in M$.
By (\ref{mp1}), (\ref{mp2}) and (\ref{mp3}) we have that
\begin{multline*}
\frac{1}{\varepsilon^{2}}I_{2}(\varepsilon,\xi)
= 2\int_{\mathbb{R}^{3}}\frac{\tilde{v}_{\varepsilon,\xi}(z)}{\varepsilon^2}U(z)\chi_{r}(\varepsilon|z|)
 (1-q\tilde{v}_{\varepsilon,\xi}(z))\frac{1}{\varepsilon}\frac{\partial U(z)}{\partial z_{k}}dz+O(\varepsilon)\\
=\frac{1}{\varepsilon} \int_{\mathbb{R}^{3}}\frac{\tilde{v}_{\varepsilon,\xi}(z)}{\varepsilon^2}\chi_{r}(\varepsilon|z|)
 \frac{\partial U^2(z)}{\partial z_{k}}dz-
\frac{q}{\varepsilon}\int_{\mathbb{R}^{3}}\frac{\tilde{v}^2_{\varepsilon,\xi}(z)}{\varepsilon^2}
 \chi_{r}(\varepsilon|z|)
\frac{\partial U^2(z)}{\partial z_{k}}dz+O(\varepsilon)
 \end{multline*}
and the last integral vanishes when $\varepsilon$ goes to zero,
because $\|\tilde{v}_{\varepsilon,\xi}\|_{L^6}\le \varepsilon^2$ uniformly with respect to $\xi$.
Thus
\begin{equation*}
\frac{1}{\varepsilon^{2}}I_{2}(\varepsilon,\xi)=
\frac{1}{\varepsilon} \int_{\mathbb{R}^{3}}\frac{\tilde{v}_{\varepsilon,\xi}(z)}{\varepsilon^2}\chi_{r}(\varepsilon|z|)
 \frac{\partial U^2(z)}{\partial z_{k}}dz+o(1)
\end{equation*}
and we can conclude exactly as for $I_1(\varepsilon,\xi)$.
\end{proof}
\begin{rem}
For every $\varphi\in H_{g}^{1}$ we have
\begin{eqnarray*}
\|\varphi\|_{H_{g}^{1}}^{2} & = &
\int_{M}|\nabla_{g}\varphi|^{2}+\varphi^{2}d\mu_{g}\le\int_{M}|\nabla_{g}\varphi|^{2}
+\frac{1}{\varepsilon}\varphi^{2}d\mu_{g}=\\
 & = & \varepsilon\int_{M}\frac{1}{\varepsilon}|\nabla_{g}\varphi|^{2}
 +\frac{1}{\varepsilon^{2}}\varphi^{2}d\mu_{g}=\varepsilon\|\varphi\|_{\varepsilon}^{2}
\end{eqnarray*}
\end{rem}

\begin{proof}[Proof of Proposition \ref{espa}]
It follows from Lemma \ref{fine4}, Lemma \ref{l5} and Lemma \ref{l5bis}.\end{proof}

\appendix

\section{Some  key estimates}

\begin{rem}
\label{remark:Weps}The following limits hold uniformly with respect
to $q\in M$.
\[
\lim_{\varepsilon\rightarrow0}\frac{1}{\varepsilon^{3}}\left|W_{\varepsilon,q}\right|_{p,g}^{p}=|U|_{p}^{p},\ \ \
2\le p\le 2^*
\]
\[
\lim_{\varepsilon\rightarrow0}\frac{1}{\varepsilon}\left|\nabla_{g}W_{\varepsilon,q}\right|_{2,g}^{2}=|\nabla U|_{2}^{2}
\]
\end{rem}
\begin{lem}
\label{lem:e3}For any $\varphi\in H_{g}^{1}(M)$ and for all $\xi\in M$
it holds
\begin{eqnarray*}
\|\Psi(W_{\varepsilon,\xi}+\varphi)\|_{H_{g}^{1}} & \le & c_{1}\varepsilon^{5/2}(1+\|\varphi\|_{\varepsilon}^{2})\\
\|\Psi(W_{\varepsilon,\xi}+\varphi)\|_{H_{g}^{1}} & \le & c_{1}(\varepsilon^{5/2}+\|\varphi\|_{H_{g}^{1}}^{2})\\
\|\Psi(W_{\varepsilon,\xi}+\varphi)\|_{L^{\infty}} & \le & c_{2}(\varepsilon^{3/2}+\|\varphi\|_{H_{g}^{1}}^{2})
\end{eqnarray*}
where $c_{1}$ and $c_{2}$ are constants non depending on $\xi$
and $\varepsilon$.\end{lem}
\begin{proof}
To simplify the notations we set $v=\Psi(W_{\varepsilon,\xi}+\varphi)$.
By (\ref{eq:e1}) we have
\begin{eqnarray*}
\|v\|_{H_{q}^{1}}^{2} & \le & \int_{M}|\nabla_{g}v|^{2}+v^{2}+q^{2}(W_{\varepsilon,\xi}+\varphi)^{2}v^{2}=q\int(W_{\varepsilon,\xi}+\varphi)^{2}v\le\\
 & \le & \left(\int_{M}v^{6}\right)^{1/6}\left(\int_{M}(W_{\varepsilon,\xi}+\varphi)^{12/5}\right)^{5/6}\le c\|v\|_{H_{g}^{1}}\|W_{\varepsilon,\xi}+\varphi\|_{L_{g}^{12/5}}^{2}\le\\
 & \le & c\|v\|_{H_{g}^{1}}\left(\|W_{\varepsilon,\xi}\|_{L_{g}^{12/5}}^{2}+\|\varphi\|_{L_{g}^{12/5}}^{2}\right)
\end{eqnarray*}
We recall (see Remark \ref{remark:Weps}) that
\begin{equation}
\lim_{\varepsilon\rightarrow0}\frac{1}{\varepsilon^{3}}\left|W_{\varepsilon,\xi}\right|_{t}^{t}=|U|_{t}^{t}\text{ uniformly w.r.t. }\xi\in M.\label{eq:e5}
\end{equation}
Then we have

\begin{equation}
\|v\|_{H_{g}^{1}}\le  c(\varepsilon^{5/2}+|\varphi|_{12/5,g}^{2})
\le  c(\varepsilon^{5/2}+\|\varphi\|_{H^1_g}^{2}).\label{eq:e6bis}
\end{equation}
Also,
\begin{equation}
\|v\|_{H_{g}^{1}}\le c\varepsilon^{5/2}(1+|\varphi|_{12/5,\varepsilon}^{2})\le
c\varepsilon^{5/2}\left(1+\|\varphi\|_{\varepsilon}^{2}\right).\label{eq:e6}
\end{equation}

By (\ref{eq:e1}) and by standard regularity theory (see \cite[Th. 8.8]{GT}),
we have that $v\in H_{g}^{2}$, and that $\|v\|_{H_{g}^{2}}\le\|v\|_{H_{g}^{1}}+\|(W_{\varepsilon,\xi}+\varphi)^{2}(1-qv)\|_{L_{g}^{2}}$.
By Sobolev embedding and by (\ref{psipos}) and (\ref{eq:e6}) we get
\begin{eqnarray}
\|v\|_{L^{\infty}} & \le & c\|v\|_{H_{g}^{2}}\le
c\left\{ \|v\|_{H_{g}^{1}}+\|(W_{\varepsilon,\xi}+\varphi)^{2}(1-qv)\|_{L^{2}}\right\} \le\nonumber \\
 & \le & c\left\{ \|v\|_{H_{g}^{1}}+\|W_{\varepsilon,\xi}\|_{L^{4}}^{2}+\|\varphi\|_{L^{4}}^{2}\right\} \le\nonumber \\
 & \le & c\left\{ \varepsilon^{3/2}+\|\varphi\|_{H_{g}^{1}}^{2}\right\} \label{eq:e7}
\end{eqnarray}
\end{proof}
\begin{lem}
\label{lem:e7}For any $\xi\in M$ and $h,k\in H_{g}^{1}$
it holds
\[
\|\Psi'(W_{\varepsilon,\xi}+k)[h]\|_{H_{g}^{1}}\le c\left\{ \varepsilon^{2}\|h\|_{H_{g}^{1}}+\|h\|_{H_{g}^{1}}\|k\|_{H_{g}^{1}}\right\}
\]
where the constant $c$ does not depend on $\xi$ and $\varepsilon$.\end{lem}
\begin{proof}
We have, by (\ref{psipos}) and by (\ref{eq:e2})
\begin{eqnarray*}
\|\Psi'(W_{\varepsilon,\xi}+k)[h]\|_{H_{g}^{1}}^{2} & = &
2q\int_{M}(W_{\varepsilon,\xi}+k)(1-q\Psi(W_{\varepsilon,\xi}+k))h\Psi'(W_{\varepsilon,\xi}+k)[h]-\\
 &  & -q^{2}\int_{M}(W_{\varepsilon,\xi}+k)^{2}(\Psi'(W_{\varepsilon,\xi}+k)[h])^{2}\le\\
 & \le & \int_{M}W_{\varepsilon,\xi}|h|\left|\Psi'(W_{\varepsilon,\xi}+k)[h]\right|+
 \int_{M}|k||h|\left|\Psi'(W_{\varepsilon,\xi}+k)[h]\right|
\end{eqnarray*}
We call each integral term respectively $I_{1},I_{2,}$
and we estimate each term separately. We have
\[
I_{1}\le\|\Psi'(W_{\varepsilon,\xi}+k)[h]\|_{L_{g}^{6}}\|h\|_{L_{g}^{6}}\|W_{\varepsilon,\xi}\|_{L_{g}^{3/2}}\le\varepsilon^{2}\|\Psi'\|_{H_{g}^{1}}\|h\|_{H_{g}^{1}}
\]
\[
I_{2}\le\|k\|_{L_{g}^{3}}\|h\|_{L_{g}^{3}}\|\Psi'(W_{\varepsilon,\xi}+k)[h]\|_{L_{g}^{3}}\le\|k\|_{H_{g}^{1}}\|h\|_{H_{g}^{1}}\|\Psi'\|_{H_{g}^{1}}
\]
that is our claim.\end{proof}
\begin{lem}
\label{lem:e5}Let us consider the functions
\[
\tilde{v}_{\varepsilon,\xi}(z)=\left\{ \begin{array}{cl}
\Psi(W_{\varepsilon,\xi})\left(\exp_{\xi}(\varepsilon z)\right) & \text{ for }z\in B(0,r/\varepsilon)\\
\\
0 & \text{ for }z\in\mathbb{R}^{3}\smallsetminus B(0,r/\varepsilon)
\end{array}\right.
\]
Then there exists a constant $c>0$ such that
\[
\|\tilde{v}_{\varepsilon,\xi}(z)\|_{L^{6}(\mathbb{R}^{3})}\le c\varepsilon^{2}.
\]
Furthermore, up to subsequences, ${\displaystyle \left\{ \frac{1}{\varepsilon^{2}}\tilde{v}_{\varepsilon,\xi}\right\} _{\varepsilon}}$
converges weakly in $L^{6}(\mathbb{R}^{3})$ as $\varepsilon$ goes
to $0$ to a function $\gamma\in D^{1,2}(\mathbb{R}^{3})$.
The function $\gamma$ solves, in a weak
sense, the equation
\begin{equation}
-\Delta\gamma=qU^{2}\text{ in }\mathbb{R}^{3}\label{eq:egamma-sol}
\end{equation}
\end{lem}
\begin{proof}
By definition of $\tilde{v}_{\varepsilon,\xi}(z)$ and by (\ref{eq:e1})
we have, for all $z\in B(0,r/\varepsilon)$,
\begin{multline}
-\sum_{ij}\partial_{j}\left(|g_{\xi}(\varepsilon z)|^{1/2}g_{\xi}^{ij}(\varepsilon z)
\partial_{i}\tilde{v}_{\varepsilon,\xi}(z)\right)=\\
=\varepsilon^{2}|g_{\xi}(\varepsilon z)|^{1/2}\left\{ qU^{2}(z)\chi_{r}^{2}(\varepsilon |z|)-\left[1+q^{2}U^{2}(z)
\chi_{r}^{2}(\varepsilon |z|)\right]\tilde{v}_{\varepsilon,\xi}(z)\right\} \label{eq:e11-1}
\end{multline}
By (\ref{eq:e11-1}), and remarking that $\tilde{v}_{\varepsilon,\xi}(z)\ge0$
we have
\begin{multline}
\|\tilde{v}_{\varepsilon,\xi}(z)\|_{D^{1,2}\left(B(0,r/\varepsilon)\right)}^{2}
\le C\int_{B(0,r/\varepsilon)}|g_{\xi}(\varepsilon z)|^{1/2}g_{\xi}^{ij}(\varepsilon z)
\partial_{i}\tilde{v}_{\varepsilon,\xi}(z)\partial_{j}\tilde{v}_{\varepsilon,\xi}(z)dz=\\
=C\varepsilon^{2}\int_{B(0,r/\varepsilon)}|g_{\xi}(\varepsilon z)|^{1/2}
\left\{ qU^{2}(z)\chi_{r}^{2}(\varepsilon |z|)\tilde{v}_{\varepsilon,\xi}(z)-
\left[1+q^{2}U^{2}(z)\chi_{r}^{2}(\varepsilon |z|)\right]\tilde{v}_{\varepsilon,\xi}^2(z)\right\} dz\le\\
\le C\varepsilon^{2}\int_{B(0,r/\varepsilon)}|g_{\xi}(\varepsilon z)|^{1/2}
qU^{2}(z)\chi_{r}^{2}(\varepsilon |z|)\tilde{v}_{\varepsilon,\xi}(z)dz\le\\
\le C\varepsilon^{2}q
\|\tilde{v}_{\varepsilon,\xi}(z)\|_{L^{6}\left(B(0,r/\varepsilon)\right)}
\|U\|_{L^{12/5}}^{2}
\le C\varepsilon^{2}\|\tilde{v}_{\varepsilon,\xi}(z)\|_{D^{1,2}\left(B(0,r/\varepsilon)\right)}\label{eq:e12}
\end{multline}
Thus we have
\begin{equation}
\|\tilde{v}_{\varepsilon,\xi}(z)\|_{D^{1,2}\left(B(0,r/\varepsilon)\right)}
\le C\varepsilon^{2}\text{ and }\|\tilde{v}_{\varepsilon,\xi}(z)\|_{L^{6}(\mathbb{R}^{3})}
\le C\varepsilon^{2}.\label{eq:e13}
\end{equation}
By (\ref{eq:e13}), if $\varepsilon_{n}$ is a sequence which goes
to zero, the sequence $\left\{ \frac{1}{\varepsilon_{n}^{2}}\tilde{v}_{\varepsilon_{n},\xi}\right\} _{n}$
is bounded in $L^{6}(\mathbb{R}^{3})$. Then, up to subsequence,
$\left\{ \frac{1}{\varepsilon_{n}^{2}}\tilde{v}_{\varepsilon_{n},\xi}\right\} _{n}$
converges to some $\tilde{\gamma}\in L^{6}(\mathbb{R}^{3})$ weakly
in $L^{6}(\mathbb{R}^{3})$.
We have also that $\|\tilde{v}_{\varepsilon,\xi}\|_{L^{2}(\mathbb{R}^{3})}\le C\varepsilon$.
In fact, by Holder inequality
\[
\int_{B(0,r/\varepsilon)}\tilde{v}_{\varepsilon,\xi}^{2}
\le\left(\int_{B(0,r/\varepsilon)}\tilde{v}_{\varepsilon,\xi}^{6}\right)^{1/3}
\left(\int_{B(0,r/\varepsilon)}1\right)^{2/3}
\le C\varepsilon^{4}\left(\frac{r^{3}}{\varepsilon^{3}}\right)^{2/3}\le C\varepsilon^{2}.
\]

Moreover, by (\ref{eq:e11-1}), for any $\varphi\in C_{0}^{\infty}(\mathbb{R}^{3})$,
it holds
\begin{multline}
\int_{\text{supp }\varphi}\sum_{ij}| g_{\xi}(\varepsilon z) |^{1/2}g_{\xi}^{ij}(\varepsilon z)
\partial_{i}\frac{\tilde{v}_{\varepsilon,\xi}(z)}{\varepsilon_{n}^{2}}\partial_{j}\varphi(z)dz=\\
\int_{\text{supp }\varphi}
\left\{ qU^{2}(z)\chi_{r}^{2}(\varepsilon |z|)-
\left[1+q^{2}U^{2}(z)\chi_{r}^{2}(\varepsilon |z|)\right]
\tilde{v}_{\varepsilon,\xi}(z)\right\} |g_{\xi}(\varepsilon z)|^{1/2}\varphi(z)dz.\label{eq:213bis}
\end{multline}
Consider now the functions
\[
v_{\varepsilon,\xi}(z):=\Psi(W_{\varepsilon,\xi})\left(\exp_{\xi}(\varepsilon z)\right)
\chi_{r}(\varepsilon |z|)=\tilde{v}_{\varepsilon,\xi}(z)\chi_{r}(\varepsilon |z|)\text{ for }z\in\mathbb{R}^{3}.
\]
We have that
\begin{eqnarray*}
\|v_{\varepsilon,\xi}(z)\|_{D^{1,2}(\mathbb{R}^{3})}^{2} & = &
\int|\nabla v_{\varepsilon,\xi}|^{2}dz\le2\int\chi_{r}^{2}(\varepsilon |z|)|\nabla\tilde{v}_{\varepsilon,\xi}(z)|^{2}
+\varepsilon^{2}|\chi'_{r}(\varepsilon |z|)|^{2}\tilde{v}_{\varepsilon,\xi}(z)^{2}dz\le\\
 & \le & c\left(\|\tilde{v}_{\varepsilon,\xi}(z)\|_{D^{1,2}\left(B(0,r/\varepsilon)\right)}^{2}
 +\varepsilon^{2}\|\tilde{v}_{\varepsilon,\xi}(z)\|_{L^{2}(\mathbb{R}^{3})}^{2}\right)\le c\varepsilon^{4}.
\end{eqnarray*}
Thus the sequence $\left\{ \frac{1}{\varepsilon_{n}^{2}}v_{\varepsilon_{n},\xi}\right\} _{n}$
converges to some $\gamma\in D^{1,2}(\mathbb{R}^{3})$ weakly in $D^{1,2}(\mathbb{R}^{3})$
and in $L^{6}(\mathbb{R}^{3})$.

For any compact set $K\subset\mathbb{R}^{3}$ eventually $v_{\varepsilon_{n},\xi}\equiv\tilde{v}_{\varepsilon_{n},\xi}$
on $K$. So it is easy to see that $\tilde{\gamma}=\gamma$.

We recall the Taylor expansions
\begin{eqnarray}
|g_{\xi}(\varepsilon z)|^{1/2}=1+O(\varepsilon^{2}|z|^2), & \text{ and } &
g_{\xi}^{ij}(\varepsilon z)=\delta_{ij}+O(\varepsilon^{2}|z|^2),\label{eq:e14}
\end{eqnarray}
so, by (\ref{eq:e14}), and by the weak convergence of
$\left\{ \frac{1}{\varepsilon_{n}^{2}}v_{\varepsilon_{n},\xi}\right\} _{n}$
in $D^{1,2}(\mathbb{R}^{3})$, for any $\varphi\in C_{0}^{\infty}(\mathbb{R}^{3})$
we get
\begin{multline}
\int_{\text{supp }\varphi}\sum_{ij}|g_{\xi}(\varepsilon_{n}z)|^{1/2}
g_{\xi}^{ij}(\varepsilon_{n}z)
\partial_{i}\frac{\tilde{v}_{\varepsilon_{n},\xi}(z)}{\varepsilon_{n}^{2}}\partial_{j}\varphi(z)dz\\
=\int_{\text{supp }\varphi}\sum_{ij}|g_{\xi}(\varepsilon_{n}z)|^{1/2}g_{\xi}^{ij}(\varepsilon_{n}z)
\partial_{i}\frac{v_{\varepsilon_{n},\xi}(z)}{\varepsilon_{n}^{2}}\partial_{j}\varphi(z)dz\\
\rightarrow\int_{\mathbb{R}^{3}}\sum_{i}\partial_{i}\gamma(z)\partial_{i}\varphi(z)dz
\text{ as }n\rightarrow\infty.\label{eq:e14bis}
\end{multline}
Thus by (\ref{eq:213bis}) and by (\ref{eq:e14bis}) and because
$\left\{ \frac{1}{\varepsilon_{n}^{2}}\tilde{v}_{\varepsilon_{n},\xi}\right\} _{n}$
converges to $\gamma$ weakly in $L^{6}(\mathbb{R}^{3})$ we get
\[
\int_{\mathbb{R}^{3}}\sum_{i}\partial_{i}\gamma(z)\partial_{i}\varphi(z)dz=
q\int_{\mathbb{R}^{3}}U^{2}(z)\varphi(z)dz\text{ for all }\varphi\in C_{0}^{\infty}(\mathbb{R}^{3}).
\]
Thus, up to subsequences,
${\displaystyle \left\{ \frac{1}{\varepsilon_{n}^{2}}\tilde{v}_{\varepsilon_{n},\xi}\right\} _{n}}$
converges to $\gamma$, weakly in $L^{6}(\mathbb{R}^{3})$
and the function $\gamma\in D^{1,2}(\mathbb{R}^{3})$ is
a weak solution of $-\Delta\gamma=qU^{2}$ in $\mathbb{R}^{3}$.\end{proof}
\begin{rem}
We remark that $\gamma$ is positive radially symmetric and decays exponentially
at infinity with its first derivative because it solves
$\displaystyle -\Delta\gamma=qU^2$ in $\mathbb{R}^3$.
\end{rem}

\end{document}